\pdfoutput=1
\documentclass{llncs}
\usepackage[utf8]{inputenc}
\usepackage[english]{babel}
\usepackage[T1]{fontenc}
\usepackage{lmodern}
\usepackage{mathtools} % Instead of amsmath
\usepackage{amssymb}
\usepackage{hyperref}
\usepackage{algorithm}
\usepackage{semantic}
\usepackage[noend]{algpseudocode}

\newcommand{\uclose}[1]{\ensuremath{{#1}^\uparrow}}
\newcommand{\dclose}[1]{\ensuremath{{#1}^\downarrow}}
\newcommand{\dclosevec}[1]{\ensuremath{\mathbf{\dclose{#1}}}}
\newcommand{\uclosure}[1]{\ensuremath{{#1}\operatorname{\uparrow}}}
\newcommand{\dclosure}[1]{\ensuremath{{#1}\operatorname{\downarrow}}}
\newcommand{\pre}[1]{\ensuremath{\mathrm{pre}(#1)}}

\newcommand{\post}[1]{\ensuremath{\mathrm{post}(#1)}}

\newcommand{\reach}{\ensuremath{\mathsf{Reach}}}
\newcommand{\cover}{\ensuremath{\mathsf{Cover}}}

\newcommand{\fire}[1]{\ensuremath{| #1 \rangle}}

\newcommand{\setR}{\ensuremath{\dclose R}}
\newcommand{\setRI}{\ensuremath{\dclose {R'}}}
\newcommand{\setP}{\ensuremath{{\dclose P}}}
\newcommand{\setU}{\ensuremath{{\uclose{\mathsf{U}}}}}
\newcommand{\setI}{I}

\newcommand{\vecF}{\ensuremath{\mathbf F}}
\newcommand{\vecR}{\ensuremath{\mathbf R}}
\newcommand{\vecRI}{\ensuremath{\mathbf R'}}

\newcommand{\en}{\ensuremath{\mathbb{N}}}
\newcommand{\NN}{\ensuremath{\mathbb{N}}}

\newcommand{\zed}{\ensuremath{\mathbb{Z}}}

\newcommand{\impl}{\ensuremath{\rightarrow}}

\newcommand{\bfa}{\mathbf{a}}
\newcommand{\bfb}{\mathbf{b}}
\newcommand{\bfc}{\mathbf{c}}
\newcommand{\bfd}{\mathbf{d}}
\newcommand{\bfg}{\mathbf{g}}
\newcommand{\bfi}{\mathbf{i}}

\newcommand{\rulename}[1]{{\ensuremath{\mathrm{[\mathsf{#1}]}}}}
\newcommand{\rInitialize}{\rulename{Initialize}}
\newcommand{\rUnfold}{\rulename{Unfold}}
\newcommand{\rModelSyn}{\rulename{ModelSyn}}
\newcommand{\rModelSem}{\rulename{ModelSem}}
\newcommand{\rValid}{\rulename{Valid}}
\newcommand{\rInduction}{\rulename{Induction}}
\newcommand{\rConflict}{\rulename{Conflict}}
\newcommand{\rCandidateUnres}{\rulename{CandidateNondet}}
\newcommand{\rCandidate}{\rulename{Candidate}}
\newcommand{\rDecideUnres}{\rulename{DecideNondet}}
\newcommand{\rDecide}{\rulename{Decide}}
\newcommand{\resValid}{\textsf{valid}}
\newcommand{\resInvalid}{\textsf{invalid}}
\newcommand{\init}{\textsf{Init}}

\newcommand{\push}[2]{\ensuremath{#1\textsc{.Push}(#2)}}
\newcommand{\popMin}[1]{\ensuremath{#1\textsc{.PopMin}}}

\newcommand{\gen}{\operatorname{Gen}}
\newcommand{\len}{\operatorname{length}}
\newcommand{\goesto}{\mapsto}

\renewcommand{\geq}{\geqslant}
\renewcommand{\leq}{\leqslant}
\renewcommand{\implies}{\Rightarrow}

\newcommand{\wqole}{\preceq}

\newcommand{\wqoge}{\succeq}

\mathtoolsset{centercolon}
\makeatletter
\def\setlabelname#1{\edef\curname{#1}\let\@currentlabelname\curname}
\makeatother

\algblock{Choose}{EndChoose}
\algcblock[Choose]{Choose}{Or}{EndChoose}

\def\set#1{{\{ #1 \}}}

\usepackage{color}

\usepackage{ntheorem}
\newtheorem*{theorem-non}{Lemma}
\newtheorem*{prop-non}{Proposition}

\begin{document}

\title{Incremental, Inductive Coverability}
\author{Johannes Kloos \and Rupak Majumdar \and Filip Niksic \and Ruzica Piskac}
\institute{MPI-SWS}
\maketitle

\begin{abstract}
We give an incremental, inductive (IC3) procedure to check coverability
of well-structured transition systems.
Our procedure generalizes the IC3 procedure for safety verification that
has been successfully applied in finite-state hardware verification to infinite-state
well-structured transition systems.
We show that our procedure is sound, complete, and terminating for {\em downward-finite}
well-structured transition systems
---where each state has a finite number of states below it---
a class that contains extensions of Petri nets, broadcast protocols, and lossy channel systems.

We have implemented our algorithm for checking coverability of Petri nets.
We describe how the algorithm can be efficiently implemented without the use of SMT solvers.
Our experiments on standard Petri net benchmarks show that IC3 is competitive with
state-of-the-art implementations for coverability based on symbolic backward analysis
or expand-enlarge-and-check algorithms both in time taken and space usage.
\end{abstract}

\section{Introduction}

The IC3 algorithm \cite{Bradley2011} was recently introduced as an efficient technique
for safety verification of hardware. 
It computes an inductive invariant by maintaining a
sequence of over-approximations of forward-reachable states,
and incrementally strengthening them based on counterexamples to inductiveness.
The counterexamples are obtained using a backward exploration from error states. 
Efficient implementations of the procedure show remarkably good performance
on hardware benchmarks \cite{Een2011}.

A natural direction is to extend the IC3 algorithm to classes of systems beyond finite-state
hardware circuits. 
Indeed, an IC3-like technique was recently
proposed for interpolation-based software verification \cite{Cimatti}, and
the technique was generalized to finite-data pushdown systems as well as systems 
using linear real arithmetic \cite{Hoder2012}.
% In general, the safety verification problem for infinite-state systems,
% such as systems defined using linear arithmetic, is undecidable.
Hoder and Bj{\o}rner show that their generalized IC3 procedure
terminates on timed pushdown automata \cite{Hoder2012}, and it is natural
to ask for what other classes of infinite-state systems does IC3 form a decision procedure
for safety verification. 

In this paper, we consider well-structured transition systems (WSTS) \cite{ACJT96,FS01}.
WSTS are infinite-state transition systems whose states have a well-quasi order, and whose
transitions satisfy a monotonicity property w.r.t.\ the quasi-order.
WSTS capture many important infinite-state models such as Petri nets and their monotonic extensions \cite{EsparzaNielsen94,Cia94,DFS98,EEC},
broadcast protocols \cite{EN98,EFM99}, and lossy channel systems \cite{ABJ98}.
A general decidability result shows that the 
coverability problem (reachability in an upward-closed set) 
is decidable for WSTS \cite{ACJT96}.
The decidability result performs a backward reachability analysis, and shows, using
properties of well-quasi orderings, that the reachability procedure must terminate.
In many verification problems, techniques based on computing inductive invariants outperform
methods based on backward or forward reachability analysis; indeed, IC3 for hardware circuits
is a prime example.
Thus, it is natural to ask if there is a IC3-style decision procedure for coverability analysis for 
WSTS.

We answer this question positively.
We give a generalization of IC3 for WSTS, and show that it
terminates on the class of {\em downward-finite} WSTS, in which each state has a finite
number of states lower than itself.
The class of downward-finite WSTS contains the most important classes of WSTS
used in verification, including 
Petri nets and their extensions, broadcast protocols, and
lossy channel systems. 
Hence, our results show that IC3 is a decision procedure for the coverability problem for
these classes of systems.
While termination is trivial in the finite-state case, our technical contribution is to show,
using the termination of the backward reachability procedure, that the sequence of (downward closed)
invariants produced by IC3 is guaranteed to converge.
We also show that the assumption of downward-finiteness is necessary: we give a (technical) example of a general WSTS
on which the algorithm does not terminate.

We have implemented our algorithm in a tool called IIC to check coverability in Petri nets.
Using combinatorial properties of Petri nets, we derive
an optimized implementation of the algorithm that does not
use an SMT solver.
Our implementation shows that IIC outperforms several state-of-the-art implementations of
coverability \cite{EEC,wahlkroening} on a set of Petri net examples, 
both in space and in time requirements.
For example, on a set of standard Petri net examples, we outperform implementations of EEC
and backward reachability, often by orders of magnitude.

\section{Preliminaries}\label{sec:prelim}

\noindent\emph{Well-quasi Orders}
For a set $X$, a relation $\wqole \subseteq X\times X$ is a
\emph{well-quasi-order (wqo)} if it is reflexive, transitive, and if
for every infinite sequence $x_0, x_1, \ldots$ of elements from $X$,
there exists $i<j$ such that $x_i\wqole x_j$.
A set $Y \subseteq X$ is \emph{upward-closed} if for every $y\in Y$
and $x\in X$, $y \wqole x$ implies $x \in Y$.  Similarly, a
set $Y\subseteq X$ is \emph{downward-closed} if for every $y\in Y$ and
$x\in X$, $x\wqole y$ implies $x\in Y$.  For a set $Y$, by
$\uclosure{Y}$ we denote its upward closure, i.e., the set $\{ x \mid
\exists y \in Y, y \wqole x \}$.  For a singleton $\set{x}$, we simply
write $\uclosure{x}$ for $\uclosure{\set{x}}$.  Similarly, we define
$\dclosure{Y} = \{ x \mid \exists y \in Y, x \wqole y \}$ for the
downward closure of a set $Y$.  Clearly, $\uclosure{Y}$ (resp.,
$\dclosure{Y}$) is an upward-closed set (resp.\ downward-closed) for
each $Y$.  The union and intersection of upward-closed sets are
upward-closed, and the union and intersection of downward-closed sets
are downward-closed. Furthermore, the complement of an upward-closed
set is downward-closed, and the complement of a downward-closed set is
upward-closed.  For the convenience of the reader, we will mark
upward-closed sets with a small up-arrow superscript, like this:
$\uclose{U}$, and downward-closed sets with a small down-arrow
superscript, like this: $\dclose{D}$.

A basis of an upward-closed set $Y$ is a set $Y_b\subseteq Y$ such
that $Y = \bigcup_{y \in Y_b} \ \uclosure{y}$.  It is known
\cite{Higman1952,ACJT96,FS01}
that any upward-closed set $Y$ in a wqo has a finite basis: the set of minimal
elements of $Y$ has finitely many equivalence classes under the
equivalence relation $\wqole \cap \wqoge$, so take any system of
representatives.  We write $\min Y$ for such a system of representatives.
Moreover, it is known that any non-decreasing
sequence $I_0\subseteq I_1\subseteq \ldots$ of upward-closed sets
eventually stabilizes, i.e., there exists $k\in \NN$ such that $I_k =
I_{k+1} = I_{k+2} = \ldots$.

A wqo $(X,\wqole)$ is \emph{downward-finite} if for each $x\in X$, the
downward closure $\dclosure{x}$ is a finite set.

\paragraph{Examples:}
Let $\NN^k$ be the set of $k$-tuples of natural numbers, and let
$\wqole$ be pointwise comparison: $v\wqole v'$ if $v_i \leq v'_i$ for $i =
1,\ldots, k$. Then, $(\NN^k, \wqole)$ is a downward-finite
wqo \cite{Dickson1913}.
%  and
% it is downward-finite, since there are $\prod_{i=1}^k (v_i+1)$
% elements below a vector $v$.

Let $A$ be a finite alphabet, and consider the subword ordering
$\preceq$ on words over $A$, given by $w \preceq w'$ for $w,w'\in
A^{*}$ if $w$ results from $w'$ by deleting some occurrences of
symbols.  Then $(A^{*}, \preceq)$ is a downward-finite wqo \cite{Higman1952}.
% and
% it is downward-finite as a fixed word has only a finite number of
% subwords.

% Also, the product of downward-finite wqos is a wqo, and the subword
% ordering on a downward-finite well-quasi ordered alphabet is a downward-finite wqo.

\smallskip
\noindent
\emph{Well-structured Transition Systems}
A well-structured transition system (WSTS) $(\Sigma, I, \to, \wqole)$
consists of a set $\Sigma$ of states, a finite set $I\subseteq \Sigma$
of initial states, a transition relation $\to \subseteq \Sigma\times
\Sigma$, and a well-quasi ordering $\wqole\subseteq \Sigma\times
\Sigma$ such that for all $s_1, s_2, t_1 \in \Sigma$ such that $s_1
\to s_2$ and $s_1 \wqole t_1$ there exists $t_2$ such that $t_1
{\to}^{*} t_2$ and $s_2 \wqole t_2$.  A WSTS is downward-finite if
$(\Sigma,\wqole)$ is downward-finite.

Let $x,y \in \Sigma$. If $x \to y$, we call $x$ a \emph{predecessor}
of $y$, and $y$ a \emph{successor} of $x$.  We write $\pre{x} := \{ y
\mid y \to x \}$ for the \emph{set of predecessors} of $x$, and
$\post{x} := \{ y \mid x \to y \}$ for the \emph{set of successors} of
$x$.  For $X \subseteq \Sigma$, $\pre{X}$ and $\post{X}$ are defined
as natural extensions, i.e., $\pre{X} = \bigcup_{x \in X} \pre{x}$ and
$\post{X} = \bigcup_{x \in X} \post{x}$.

We write $x \to^k y$ if there are states $x_0, \ldots, x_k \in \Sigma$
such that $x_0 = x$, $x_k = y$ and $x_i \to x_{i+1}$ for $0 \le i <
k$.  Furthermore, $x \to^{*} y$ iff there exists a $k\geq 0$ such that
$x \to^k y$, i.e., $\to^{*}$ is the reflexive and transitive closure
of $\to$.  We say that there is a \emph{path from $x$ to $y$ of length
  $k$} if $x \to^k y$, and that there is a path from $x$ to $y$ if $x
\to^{*} y$.

The set of \emph{$k$-reachable} states $\reach_k$ is the set of states
reachable in at most $k$ steps, formally, 
$\reach_k := \{ y \in \Sigma \mid \exists k' \leq k,\exists x \in I, x \to^{k'} y \}$.
The set of \emph{reachable} states 
$\reach := \bigcup_{k \ge 0} \reach_k = \{ y \mid \exists x \in I, x \to^{*} y \}$.  
Using downward closure, we can define the \emph{$k$-th cover} $\cover_k$ and
the \emph{cover} $\cover$ of the WSTS as
$\cover_k := \dclosure{\reach_k}$ and $\cover := \dclosure{\reach}$.
The \emph{coverability problem for WSTS} asks, given a WSTS
$(\Sigma,I, \to,\wqole)$ and a downward-closed set $\setP$, if every
reachable state is contained in $\setP$, i.e., if $\reach \subseteq
\setP$.  It is easy to see that this question is equivalent to
checking if $\cover \subseteq \setP$.
% So if $\cover$ were computable, this would provide a solution to the problem.

In the following, we make some standard effectiveness assumptions on
WSTS \cite{ACJT96,FS01}.  We assume
that $\wqole$ is decidable, and that for any state $x\in\Sigma$, there
is a computable procedure that returns a finite basis for
$\pre{\uclosure{x}}$.  These assumptions are met by most classes of
WSTS considered in verification \cite{FS01}.

Under the preceding effectiveness assumptions, one can show that the
coverability problem is decidable for WSTS by a backward-search algorithm
\cite{ACJT96}.  The main construction is the
following sequence of upward-closed sets:
\begin{align} \label{abdulla-seq} \setU_0 &:= \Sigma \setminus \setP\,, &
  \setU_{i+1} & := \setU_i \cup \pre{\setU_i}\,.
\tag{BackwardReach}
\end{align}
It is easy to see that the sequence of sets $\setU_i$ forms an
increasing chain of upward-closed sets, therefore it
eventually stabilizes: there is some $L$ such that $\setU_L =
\setU_{L+i}$ for all $i \ge 0$.  Then, $\cover \subseteq \setP$ iff $I
\cap \setU_L = \emptyset$.  Moreover, if $I \cap \setU_L = \emptyset$,
then $\Sigma \setminus \setU_L$ contains $I$, is contained in $\setP$
and satisfies $\post{\Sigma\setminus\setU_L}\subseteq\Sigma\setminus\setU_L$.

% In general, $\cover$ need not be a computable set, even when the computability assumptions on WSTS hold.
%
We generalize from $\Sigma\setminus \setU_L$, in the style of
inductive invariants for safety verification, to the notion of an
(inductive) \emph{covering set}.  A downward-closed set $\dclose{C}$ is
called a \emph{covering set} for $\setP$ iff (a) $I\subseteq\dclose{C}$, (b) $\dclose{C}
\subseteq \setP$, and (c) if $\post{\dclose{C}} \subseteq \dclose{C}$.  By
induction, it is clear that $\cover \subseteq \dclose{C} \subseteq
\setP$ for any covering set $\dclose{C}$.  
Therefore, to solve the coverability
problem, it is sufficient to exhibit any covering set.
%  such that $\dclose{C} \subseteq \setP$,
% or demonstrate that such a set cannot exist. In fact, to demonstrate
% that no such covering set can exist, it is sufficient to show that some
% state $x \in \reach$ such that $x \not\in \setP$, i.e., there is a path
% from some state in $I$ to some state not in $\setP$.

\section{IC3 for Coverability}\label{sec:alg}

We now describe an algorithm for the coverability problem that takes
as input a WSTS $(\Sigma, I, \to, \wqole)$ and a downward-closed set
$\setP$, and constructs
either a path from some state in $I$ to a state not in $\setP$ 
(if $\cover \not\subseteq \setP$), or an inductive covering set for $\setP$.
%
%\begin{enumerate}
%	\item A path from an initial state $x \in \setI$ to a state
%	    $y \not\in \setP$, in case $\reach \not \subseteq \setP$; or
%	\item \FN{The inductive covering and the fact it contains $\cover$
%	    is discussed in the last paragraph of preliminaries. Is it necessary
%	    to repeat all that? Maybe put (a)-(c) there explicitly.}
%	   
%	    An inductive covering set $\setR$ such that:\\
%		(a) $\setI \subseteq \setR$,\quad\quad
%		(b) $\post{\setR} \subseteq \setR$,\quad\quad
%		(c) $\setR \subseteq \setP$.\\
%	    By induction, one finds that this implies
%	    $\cover \subseteq \setR{}$.
%    \end{enumerate}
%
    In the algorithm we consider sets that are not necessarily
    inductive by themselves, but they are \emph{inductive relative to}
    some other sets. Formally, for a set $\setR$ such that $\setI
    \subseteq \setR$, a downward-closed set $\dclose{S}$ is inductive relative to
    $\setR$ if $I \subseteq \dclose{S}$ and $\post{\setR \cap \dclose{S}} \subseteq \dclose{S}$.
    An upward-closed set $\uclose{U}$ is inductive relative to $\setR$ if its downward-closed
    complement $\Sigma\setminus\uclose{U}$ is inductive relative to $\setR$, i.e.
    if $I\cap\uclose{U}=\emptyset$ and
    $\post{\setR\setminus\uclose{U}}\subseteq\Sigma\setminus\uclose{U}$.
    
It can be easily shown that the condition $\post{\setR \cap
  \dclose{S}} \subseteq \dclose{S}$ is equivalent to $\pre{\Sigma
  \setminus \dclose{S}} \cap \setR \cap \dclose{S} = \emptyset$.
  Stated in terms of an upward-closed set $\uclose{U}$, the equivalent
  condition is $\pre{\uclose{U}}\cap\setR\setminus\uclose{U}=\emptyset$.
  Since all these conditions are equivalent, we will use them interchangeably.

% In the following, we will describe the algorithm by giving a set
% of state transition rules, similar to \cite{Hoder2012}.  
% In a first step, we will provide a set of
% non-deterministic rules that we will prove sound for any WSTS: whenever the
% algorithm terminates and makes a statement about coverability,
% this statement reflects the properties of the WSTS.  
% In a second step, we resolve some non-deterministic choices and show the 
% algorithm always terminates for downward-finite WSTS.

\subsection{Algorithm}

    \begin{figure}[t]
	\begin{align*}
	    &\inference[\rInitialize]{ }
	      {\init \goesto \dclosure{\setI} \mid  \varnothing}
	    &&\inference[\rCandidateUnres]{a \in \setR_N \setminus \setP}
			 {\vecR\mid \varnothing \goesto \vecR\mid \langle a, N \rangle}\\
	    &\inference[\rModelSyn]{\min Q = \langle a,0\rangle}
		  {\vecR\mid Q \goesto \resInvalid}
	    &&\inference[\rModelSem]{\min Q = \langle a,i \rangle
		  \quad I \cap \uclosure{a} \neq \varnothing}
		  {\vecR\mid Q \goesto \resInvalid}\\
	      &\mathrlap{\inference[\rDecideUnres]
			 {\min Q = \langle a, i\rangle \quad i > 0
			     \quad b \in \pre{\uclosure{a}} \cap \setR_{i-1} \setminus \uclosure{a}}
			 {\vecR\mid Q \goesto \vecR\mid \push{Q}{\langle b,i-1\rangle}}}\\
	    &\mathrlap{\inference[\rConflict]
			{\min Q = \langle a, i \rangle
			    \quad i > 0
			    \quad \pre{\uclosure a} \cap \setR_{i-1} \setminus \uclosure{a}
			    	= \varnothing
			    \quad b \in \gen_{i-1}(a)}
			{\vecR\mid Q \goesto
			 \vecR[\setR_k \gets \setR_k \setminus
			 \uclosure{b}]_{k=1}^i\mid \popMin{Q}}}\\
	    &\mathrlap{\inference[\rInduction]
		  {\setR_i = \Sigma \setminus
		      \uclosure{ \{ r_{i,1}, \ldots, r_{i,m} \}}
		  \quad b \in \gen_i(r_{i,j}) \text{ for some } 1 \le j \le m}
		  {\vecR\mid \varnothing \goesto
		      {\vecR[\setR_k \gets \setR_k \setminus
		  \uclosure{b}]_{k=1}^{i+1}\mid \varnothing}}}\\
	    &\inference[\rValid]{\setR_i = \setR_{i+1}
		  \text{ for some } i < N}
		  {\vecR\mid Q \goesto \resValid}
	    &&\inference[\rUnfold]{\setR_N \subseteq \setP}
	      {\vecR\mid \varnothing \goesto \vecR \cdot \Sigma\mid \varnothing}
	\end{align*}
	\caption{The rule system for a IC3-style algorithm for WSTS --
	    generic version. The map $\gen_i$ is defined in
	    equation \eqref{eq:gen}.
%Compare \cite{Hoder2012}, Figure 2.
	\label{fig:rules-generic}
}
    \end{figure}

    Figure~\ref{fig:rules-generic} shows the algorithm as a set of
    non-deterministic state transition rules, similar to \cite{Hoder2012}.  
    A state of the computation is either the
    initial state $\init$, the special states $\resValid{}$ and
    $\resInvalid{}$ that denote termination, or a pair $\vecR \mid Q$
    defined as follows.

    The first component of the pair is a vector $\vecR$ of
    downward-closed sets, indexed starting from 0.  The elements of
    $\vecR$ are denoted $\setR_i$.  In particular, we denote by
    $\setR_0$ the downward closure of $\setI$, i.e., $\setR_0 =
    \dclosure{\setI}$.  These sets contain the successive
    approximations to the inductive covering set.  The function
    $\operatorname{length}$ gives the length of the vector,
    disregarding $\setR_0$, i.e., $\operatorname{length}(\setR_0,
    \ldots, \setR_N) = N$.  If it is clear from the context which
    vector is meant, we often abbreviate
    $\operatorname{length}(\vecR)$ simply with $N$.  We write $\vecR
    \cdot X$ for the concatenation of the vector $\vecR$ with the
    downward closed set $X$: $(\setR_0, \ldots, \setR_N) \cdot X =
    (\setR_0, \ldots, \setR_N, X)$.

    The second component of the pair is a priority queue $Q$,
    containing elements of the form $\langle a, i\rangle$, where $a
    \in \Sigma$ is a state and $i\in \NN$ is a natural number.  The
    priority of the element is given by $i$, and is called the level
    of the element.  The statement $\langle a, i \rangle \in Q$ means
    that the priority queue contains an element of the given form,
    while $\min Q = \langle a, i\rangle$ means that the minimal
    element of the priority queue has the given form. Furthermore,
    $\popMin{Q}$ yields $Q$ after removal of its minimal element, and
    $\push{Q}{x}$ yields $Q$ after adding element $x$.

    The elements of $Q$ are states that lead outside of $\setP$.  Let
    $\langle a, i\rangle$ be an element of $Q$. Either $a$ is a state
    that is in $R_i$ and outside of $\setP$, or there is a state
    $b$ leading to $\setP^c$ such that $a \in
    \pre{\uclosure{b}}$. Our goal is to try to discard those states
    and show that they are not reachable from the initial state, as
    $R_i$ denotes an overapproximation of the states reachable in $i$
    or less steps. If an element of $Q$ is reachable from the initial
    state, then $\cover \not \subseteq \setP$.

    The state \resValid{} signifies that the search has terminated
    with the result that $\cover \subseteq \setP$ holds, while
    \resInvalid{} signifies that the algorithm has terminated with the
    result that $\cover \not\subseteq \setP$.  In the description of
    the algorithm, we will omit the actual construction of
    certificates and instead just state that the algorithm terminates
    with $\resInvalid$ or $\resValid$; the calculation of certificates is
    straightforward.

    The transition rules of the algorithm are of the form
\begin{equation}\label{rule}
	\inference[{[Name]}]{C_1 \quad \cdots C_k}
	{\sigma \goesto \sigma'}
\tag{Rule}
\end{equation}
and can be read thus: whenever the algorithm is in state $\sigma$ and
conditions $C_1$--$C_k$ are fulfilled, the algorithm can apply rule
[Name] and transition to state $\sigma'$.  We write $\sigma \goesto
\sigma'$ if there is some rule such that the algorithm applies the
rule to go from $\sigma$ to $\sigma'$.  We write $\goesto^{*}$ for the
reflexive transitive closure of $\goesto$.

Let $\mathsf{Inv}$ be a predicate on states.  When we say that a rule
\emph{preserves the invariant} $\mathsf{Inv}$ if whenever
$\sigma$ satisfies $\mathsf{Inv}$ and conditions $C_1$ to $C_k$ are
met, it also holds that $\sigma'$  satisfies $\mathsf{Inv}$.

Two of the rules use the map $\operatorname{Gen}_i: \Sigma \to
2^{\Sigma}$. It yields those states that are valid
generalizations of $a$ relative to some set $\setR_i$.  A state $b$ is
a generalization of the state $a$ relative to the set $\setR_i$,
if $b \wqole a$ and $\uclosure{b}$ is inductive relative to
$\setR_i$. Formally,
\begin{equation}
	\label{eq:gen}
	\operatorname{Gen}_i(a) := \{ b \mid b \wqole a \land
	    	\uclosure{b} \cap \setI = \varnothing \land
                 \pre{\uclosure{b}} \cap \setR_{i} \setminus \uclosure{b} = \varnothing \}
\end{equation}
Finally, the notation $\vecR[\setR_k \gets \setRI_k]_{k=1}^i$ means
that $\vecR$ is transformed by replacing 
$\setR_k$ by $\setRI_k$ for each $k = 1, \ldots, i$, i.e.,
\[ \vecR[\setR_k \gets \setRI_k]_{k=1}^i
	= (\setR_0, \setRI_1, \ldots, \setRI_i,	\setR_{i+1}, \ldots, \setR_n).
\]
We provide an overview of each rule of the calculus.
%\begin{itemize}
%\item 
$\rInitialize$ The algorithm starts by defining the first
  downward-closed set $\setR_0$ to be the downward closure of the initial
  state.
\\
%\item 
$\rCandidateUnres$ If there is a state $a$ such that $a \in \setR_N$
  but at the same time it is not an element of $\setP$ we add $\langle
  a, N\rangle$ to the priority queue $Q$.
\\
%\item 
$\rDecideUnres$ To check if the elements of $Q$ are spurious
  counterexamples, we start by processing an element $a$ with the
  lowest level $i$. If there is an element $b$ in $\setR_{i-1}$ such that
  $b \in \pre{\uclosure{a}}$, then we add $\langle b, i-1\rangle$ to
  the priority queue $Q$.
\\
%\item 
$\rModelSyn$ If the queue contains a state $a$ from the level 0,
  then we have found a counterexample trace and the algorithm
  terminates in the state $\resInvalid$.
\\
%\item 
$\rModelSem$ Similarly, if the queue contains a state $a$ such
  $I \cap \uclosure{a} \neq \emptyset$, this is again a counterexample
  trace and the algorithm terminates in the state $\resInvalid$.
\\
%\item 
$\rConflict$ If none of predecessors of a state $a$ from the
  level $i$ is contained in $\setR_{i-1} \setminus \uclosure{a}$, then
  $a$ belongs to a spurious counterexample trace and can be safely
  removed from the queue. Additionally, we update the downward-closed
  sets $\setR_1, \ldots, \setR_i$ as follows: since the states in
  $\uclosure{a}$ are not reachable in $i$ steps, they can be safely
  removed from all the sets $\setR_1, \ldots, \setR_i$.  Moreover,
  instead of $\uclosure{a}$ we can remove even a bigger set
  $\uclosure{b}$, for any state $b$ which is a generalization of the
  state $a$ relative to $\setR_{i-1}$, as defined in \eqref{eq:gen}.
\\
%\item 
$\rInduction$ If for some state $\uclosure{r_{i,j}}$ that was
  previously removed from $\setR_i$, a set
  $\Sigma\setminus\uclosure{r_{i,j}}$ becomes inductive relative to
  $\setR_i$ (i.e.  $\post{\setR_i \cap \dclosure{r_{i,j}}} \subseteq
  \dclosure{r_{i,j}}$), none of the states in $\uclosure{r_{i,j}}$ can
  be reached in at most $i+1$ steps. Thus, we can safely remove
  $\uclosure{r_{i,j}}$ from $\setR_{i+1}$ as well.  Similarly as in
  $\rConflict$, we can even remove $\uclose{b}$ for any generalization
  $b\in\operatorname{Gen}_i(r_{i,j})$.
\\
%\item 
$\rValid$ If there is a downward-closed set $\setR_i$ such that
  $\setR_i = \setR_{i+1}$, the algorithm terminates in the state
  $\resValid$.
\\
%\item 
$\rUnfold$ If the queue is empty and all elements of $\setR_N$
  are in $\setP$, we start with a construction of the next set
  $\setR_{N+1}$. Initially, $\setR_{N+1}$ contains all the states,
  $\setR_{N+1} = \Sigma$, and we append $\setR_{N+1}$ to the vector
  $\vecR$.
%\end{itemize}

%

\subsection{Soundness}

We first show that the algorithm is sound: if it terminates, it
produces the right answer.  If it terminates in the state
\resInvalid{} there is a path from an initial state to a state outside
of $\setP$, and if it terminates in the state \resValid{} then 
$\cover \subseteq \setP$.

We prove soundness by showing that on each state $\vecR\mid Q$ the
following invariants are preserved by the transition rules:
    \begin{align}
	& \setI \subseteq \setR_i &&\text{for all } 0 \le i \le N
	  \label{eq:I1}\tag{I1}\\
        & \post{\setR_i} \subseteq \setR_{i+1} &&\text{for all } 0 \le i < N
	  \label{eq:I2}\tag{I2}\\
	& \setR_i \subseteq \setR_{i+1} &&\text{for all } 0 \le i < N
	  \label{eq:I3}\tag{I3}\\
	& \setR_i \subseteq \setP &&\text{for all } 0 \le i < N
	  \label{eq:I4}\tag{I4}
    \end{align}
    These properties imply $\setR_i \supseteq \cover_i$, that is, the
    region $R_i$ provides an over-approximation of the $i$-cover.

    The first step of the algorithm (rule \rInitialize) results with
    the state $\dclosure{\setI}\mid \varnothing$, which satisfies
    \eqref{eq:I2}--\eqref{eq:I4} trivially, and $\setI \subseteq
    \dclosure{\setI}$ establishes \eqref{eq:I1}.  The following lemma
    states that the invariants are preserved by rules that do not
    result in \resValid{} or \resInvalid.  
For lack of space, full proofs are given in Appendix~\ref{ap:proofs}.

    \begin{lemma}\label{lemma:rulesI14invs}
      The rules $\rUnfold$, $\rInduction$, $\rConflict$, $\rCandidateUnres$,
      and $\rDecideUnres$ preserve \eqref{eq:I1} -- \eqref{eq:I4},
\end{lemma}

By induction on the length of the trace, it can be shown that if
$\init \goesto^{*} \vecR\mid Q$, then $\vecR \mid  Q$ satisfies \eqref{eq:I1}
-- \eqref{eq:I4}.  When $\init \goesto^{*} \resValid$, there is a
state $\vecR\mid Q$ such that $\init \goesto^{+} \vecR\mid Q \goesto
\resValid$, and the last applied rule is \rValid.  To be able to apply
\rValid, there is an $i$ such that $\setR_i = \setR_{i+1}$.

We claim that $\setR_i$ is an inductive covering set.  This claim
follows from the fact that (1) $\setR_i \subseteq \setP$ by invariant
\eqref{eq:I4}, (2) $\setI \subseteq \setR_i$ by invariant
\eqref{eq:I1}, and (3) $\post{\setR_i} \subseteq \setR_{i+1} =
\setR_i$ by invariant \eqref{eq:I2}.  This claim proves the
correctness of the algorithm in case $\cover \subseteq \setP$:

\begin{theorem}{[Soundness of coverability]}\label{tm:soundCover}
Given a WSTS $(\Sigma, I, \to, \wqole)$ and a downward-closed set
$\setP$, if $\init \goesto^{*} \resValid$, then $\cover \subseteq \setP$.
\end{theorem}

We next consider the case when $\cover \not\subseteq \setP$. The
following lemma describes the structure of the priority queues used in
the algorithm. 

\begin{lemma}{}
  \label{lem:Q-structure}
  Let $\init \goesto^{*} \vecR\mid Q$. 
  If $Q \neq \varnothing$, then for every $\langle a, i\rangle \in Q$,
    there is a path from $a$ to some $b \in \Sigma \setminus \setP$.
\end{lemma}

\begin{theorem}{[Soundness of uncoverability]}
Given a WSTS $(\Sigma, I, \to, \wqole)$ and a downward-closed set
$\setP$, if $\init \goesto^{*} \resInvalid$, then $\cover \not \subseteq \setP$.
\end{theorem}

\begin{proof}
  The assumption $\init \goesto^{*} \resInvalid$ implies that there is
  some state $\vecR\mid Q$ such that $\init \goesto^{*} \vecR\mid Q \goesto
  \resInvalid$, and the last applied rule was either \rModelSyn{} or
  \rModelSem. In both cases $Q \neq \emptyset$.

  If the last applied rule was \rModelSem, there is an $\langle a,
  i\rangle \in Q$ such that $\uclosure{a} \cap I \neq \emptyset$. By
  Lemma~\ref{lem:Q-structure} there is a path from $a$ to $b \in
  \Sigma \setminus \setP$. Let $a' \in \uclosure{a} \cap I$.  Since
  $(\Sigma, I, \to, \wqole)$ is a WSTS, there is $b'$ such that $a'
  \to^{*} b'$ and $b' \wqoge b$. The set $\Sigma \setminus \setP$ is
  upward-closed, and thus $b' \in \Sigma \setminus \setP$. The path
  $a' \to^{*} b'$ is a path from $I$ to $\Sigma \setminus \setP$,
  proving that $\cover \not\subseteq \setP$.

  If the last applied rule was \rModelSyn, then $\langle a, 0\rangle
  \in Q$. This implies $a \in \setR_0 = \dclosure{I}$, as $\setR_0$ is
constant in the algorithm. 
 Equivalently, $\uclosure{a} \cap I \neq \varnothing$ and 
we apply the same arguments as in the case for
  \rModelSem.  \qed
\end{proof}

\subsection{Termination}

While the above non-deterministic rules guarantee soundness for any WSTS, termination
requires some additional choices. 
We modify the $\rDecideUnres$ and $\rCandidateUnres$ rules into more
restricted rules \rDecide{} and \rCandidate, while all other rules are unchanged.
Figure~\ref{fig:rules-restricted} shows the new rules \rCandidate{} and \rDecide{}.
These rules additionally use a sequence of sets $D_i$. 
Intuitively, there can be infinitely many elements in $\setR_N \setminus \setP$. 
Sets $D_i$ provide a finite representation of those elements. 

Recall the sequence $\setU_i$ of backward reachable states from \eqref{abdulla-seq}.
We define sets $D_i$ using sets $\setU_i$. 
The set $D_i$ captures all new elements that are introduced in $\setU_i$ and that
were not present in the previous iterations.
Formally, we define sets $D_i$ as follows:
\begin{align}
  D_0 := \min (\Sigma \setminus \setP) & \quad & 
  D_{i+1} := \bigcup_{a \in D_i}\min (\pre{\uclosure{a}}) \setminus \setU_i\enspace .
\end{align}
By induction, and the finiteness of the set of minimal elements, we have that $D_i$ is finite for all $i\geq 0$.
Further, assume that $\setU_L = \setU_{L+1}$.
% This sequence eventually stabilizes and there is $L$ such that $\setU_L = \setU_{L+i}$ for all $i \ge 0$.
% We use the definition of $D_i$ to show that there are only finitely many different sets $\setR_i$.
Then, $D_i = \varnothing$ for all $i > L$.
As a consequence, the set $\bigcup_{i \geq 0} D_i$ is finite.

% \begin{lemma}
% 	\label{prop:finite_di}
% 	For $i > L$, $D_i = \varnothing$. This implies that the set $\bigcup_{i \ge 0} D_i$ is finite.
%    \end{lemma}

    \begin{figure}[t]
	\begin{align*}
	      &\inference[\rCandidate]{a \in \setR_N \cap D_0}
	    		{\vecR\mid \varnothing \goesto \vecR\mid \langle a,N \rangle}
              &&\inference[\rDecide]{\min Q = \langle a,i \rangle
			\quad i > 0 \quad b \in D_{N-i+1} \cap \setR_{i-1} \quad b \to a}
			{\vecR\mid Q \goesto \vecR\mid \push{Q}{\langle b, i-1\rangle}}
	\end{align*}
	\caption{Rules replacing \rCandidateUnres{} and
		\rDecideUnres{} in Fig.~\ref{fig:rules-generic}.
	\label{fig:rules-restricted}
        }
    \end{figure}

It is easy to show that the restricted rules still preserve the invariants
\eqref{eq:I1} -- \eqref{eq:I4}, 
and thus the modified algorithm is still sound.
From now, we focus on the modified algorithm. 
% We begin with some properties of the sequence $D_i$.
% USED IN THE APPENDIX 
% \begin{lemma}\label{lem:D_iProperties}
% Consider a WSTS $(\Sigma, I, \to, \wqole)$, a downward-closed set
% $\setP$, and the sequence of sets $D_i$. 
% Let $\init \goesto^{*} \vecR\mid Q$.
% % The sets $D_i$ satisfy the following properties:
% \begin{enumerate}
%   % \item $D_0 \subseteq \Sigma \setminus \setP$ and $D_{i+1} \subseteq \pre{D_i} \setminus U_i$.
%   \item Whenever $\setR_N \setminus \setP \neq \varnothing$, there exists an $x \in \setR_N \cap D_0$.
%   \item For all $a \in D_i$, if $\pre{\uclosure{a}} \cap \setR_{N-i-1} \neq \varnothing$, there exists an element 
% $x \in \pre{\uclosure{a}} \cap D_{i+1} \cap \setR_{N-i-1}$.
%   % \item $D_i$ is finite for all $i \ge 0$.
%   \item For all $i \ge 1$, 
% 		$\setR_i = \Sigma \setminus \{ r_{i,1}, \ldots, r_{i,m_i} \}$,
% 		where for all $j = 1, \ldots, m_i$, there is a
% 		$k \ge 0$ and a $d \in D_k$ such that $r_{i,j} \wqole d$.
%   \item For all $\langle a, i\rangle \in Q$, we have $a \in D_{N-i}$.
% \end{enumerate}
% \end{lemma}

% Moreover, in the new transition system, a set $D_i$ describes the states that \rCandidate{} and \rDecide{} will find
% to add them to the queue.
%     \begin{lemma}{}
% 	\label{lem:finite-choice}
% 	\begin{enumerate}
% 	    \item For all $i \ge 1$,
% 		$\setR_i = \Sigma \setminus \{ r_{i,1}, \ldots, r_{i,m_i} \}$,
% 		where for all $j = 1, \ldots, m_i$, there is a
% 		$k \ge 0$ and a $d \in D_k$ such that $r_{i,j} \le d$.
% 	    \item For all $\langle a, i\rangle \in Q$, $a \in D_{N-i}$.
% 	\end{enumerate}
%     \end{lemma}

To show that the algorithm always terminates, we first show that the
system can make progress until some final state is reached.

% USE THE LEMMA IN THE APPENDIX
%    \begin{lemma}[Progress]
%	\label{lem:progress}
%	If $\init \goesto^{*} \vecR\mid Q$, then either $\vecR\mid Q \goesto \vecRI\mid Q'$, or $\vecR\mid Q \goesto \resValid$, or
%	$\vecR\mid Q \goesto \resInvalid$.
%    \end{lemma}

    \begin{proposition}[Maximal finite sequences]
	\label{cor:max-finite-seqs}
	Let $\init = \sigma_0 \goesto \sigma_1 \goesto \cdots \goesto \sigma_K$
	be a maximal sequence of states, i.e., a sequence such that there is
	no $\sigma'$ such that $\sigma_K \goesto \sigma'$. Then
	$\sigma_K = \resValid$ or $\sigma_K = \resInvalid$.
    \end{proposition}

We prove the termination of the algorithm by defining a well-founded ordering
on the tuples $\vecR\mid Q$.

    \begin{definition}
	Let $\mathbf{\dclose{A}} = (\dclose{A_1}, \ldots, \dclose{A_N})$
	and $\mathbf{\dclose{B}} = (\dclose{B_1}, \ldots, \dclose{B_N})$
	be two finite sequences of downward-closed sets of the equal length $N$.
Define $\dclosevec{A} \sqsubseteq \dclosevec{B}$
if{}f $\dclose{A_i} \subseteq \dclose{B_i}$ for all $i=1,\ldots,N$.
% NOT USED IN THE MAIN TEXT:
% Write $\dclosevec{A} \sqsubset \dclosevec{B}$ if 
% $\mathbf{\dclose{A}} \sqsubseteq \mathbf{\dclose{B}}$ but
% $\mathbf{\dclose{A}} \neq \mathbf{\dclose{B}}$.
%
Let $Q$ be a priority queue whose elements are tuples $\langle a, i \rangle \in \Sigma \times \NN$, 
and let $N$ be a natural number.
Define
	$\ell_N(Q):=\min(\{ i \mid \langle a, i \rangle \in Q\}\cup\{ N+1 \})$,
to be the smallest priority in $Q$ or $N+1$ if $Q$ is empty.

For two states $\vecR\mid Q$ and $\vecRI\mid Q'$, such $\operatorname{length}(\vecR) = \operatorname{length}(\vecRI) = N$,
we define the ordering $\le_s$ as: 
\begin{align*}
	    \vecR\mid Q \le_s \vecRI\mid Q' &:\iff \vecR \sqsubseteq \vecRI
		\land (\vecR = \vecRI \impl \ell_{N}(Q) \le
		\ell_{N}(Q')
	\end{align*}
and we write $\vecR\mid Q <_s \vecRI\mid Q'$ if $\vecR\mid Q \le_s \vecRI\mid Q'$
but $\vecR \neq \vecRI$ or $Q\neq Q'$. 
\end{definition}

    \begin{lemma}[$\le_s$ is a well-founded quasi-order.]
	\label{lem:les_wfqo}
	The relation $<_s$ is a well-founded strict quasi-ordering on the set
	$(\mathcal D)^{*} \times \mathcal Q$, where
	$\mathcal D$ is a set of downward-closed sets over $\Sigma$, and
	$\mathcal Q$ denotes the set of priority queues over
	$\Sigma \times \en$.
    \end{lemma}

% USED IN THE APPENDIX:
%\begin{lemma}
%	\label{lem:rules-R-order}
%	If $\vecR\mid Q \goesto \vecRI\mid Q'$ as a result of  applying the
%	\rCandidate,  \rDecide, \rConflict, or \rInduction{} rule, then
%	$\vecR\mid Q >_s \vecRI\mid Q'$.
%    \end{lemma}

% An easy consequence of Lemma~\ref{lem:rules-R-order} is the following corollary that will play an important role
% in the proof of termination.
The following proposition characterizes infinite runs of the algorithm.
The proof follows from the observation that
if $\vecR\mid Q \goesto \vecRI\mid Q'$ as a result of  applying the
\rCandidate,  \rDecide, \rConflict, or \rInduction{} rules, then
$\vecR\mid Q >_s \vecRI\mid Q'$.

    \begin{proposition}[Infinite sequence condition]\label{prop:inner-loop-terminates}
	For every infinite sequence
	$\init \goesto \sigma_1 \goesto \sigma_2 \goesto \cdots$,
	there are infinitely many $i$ such that $\sigma_i \goesto \sigma_{i+1}$
	by applying the rule $\rUnfold$.
    \end{proposition}

We first prove that the algorithm terminates for the case when $\cover \not\subseteq \setP$.

\begin{lemma}
	\label{lem:no-bad-unfolds}
        Let $(\Sigma, I,\to,\wqole)$ be a WSTS and $\setP$ a downward-closed set
        such that $\cover_k \cap (\Sigma \setminus \setP) \not = \emptyset$. 
	For every sequence $\init \goesto \sigma_1 \goesto^{*} \sigma_n$,
	there are at most $k$ different values for $i$ such that
	$\sigma_i \goesto \sigma_{i+1}$ using the $\rUnfold$ rule.
\end{lemma}

\begin{theorem}{[Termination when $\cover \not\subseteq \setP$]}
	\label{thm:termination-reachable}
Given a WSTS $(\Sigma, I, \to, \wqole)$ and a downward-closed set
$\setP$, if $\cover \not\subseteq \setP$, then the algorithm terminates and
all maximal execution sequences have the form $\init \goesto^{*} \resInvalid$.
\end{theorem}

    \begin{proof}
	Since $\cover \not\subseteq \setP$, there is a state
	$y \in \cover \setminus \setP$. By the definition of \cover,
	there are states $x', y'$ such that $x' \in I$, $y' \wqoge y$
	and $x' \to^k y'$ for some $k \ge 0$. Because $\Sigma \setminus \setP$
	is upward-closed, we have $y' \in \Sigma \setminus \setP$. Combining
	Lemma \ref{lem:no-bad-unfolds} and Proposition
	\ref{prop:inner-loop-terminates}, we prove that the algorithm terminates.

	Let $\init \goesto^{*} \sigma$ be a maximal execution. By
	Proposition~\ref{cor:max-finite-seqs}, $\sigma = \resValid$ or
	$\sigma = \resInvalid$. By Theorem~\ref{tm:soundCover},
	$\sigma \neq \resValid$.
\qed
    \end{proof}

To prove that the algorithm terminates when $\cover \subseteq \setP$, we use an additional assumption:
\begin{verse}
Apply $\rValid$ whenever it is applicable. \hfill ($\dagger$)
\end{verse}
This is natural assumption: since the algorithm is used to decide the coverability
problem and $\rValid$ answers the problem positively, choosing the
$\rValid$ rule when it is applicable is the most efficient choice.
The main argument for showing the termination will reduce to showing that, for downward-finite WSTS, we can
generate only a finite number of different sets $\setR_i$, so $\rValid$ will be
applicable at some point. 
The key combinatorial property of downward-finite wqos is as follows.

\begin{lemma}
	\label{lem:termination}
Let $(\Sigma,\wqole)$ be a downward-finite wqo and let $D$ be a finite set.
Consider a sequence $\setR_0 \subseteq \setR_1 \subseteq \ldots$,
where each $\setR_i = \Sigma \setminus \uclosure{\set{r_{i,1},\ldots, r_{i,m_i}}}$
for $r_{i,j}\in \dclosure{D}$.
Then there is $K\in \NN$ such that $R_K = R_{K+1}$.
\end{lemma}
% There is a number $K$ such that for every vector
% $\vecR = (\setR_0, \ldots, \setR_N)$
% satisfying these three conditions:
% 	\begin{enumerate}
% 	    \item For all $1 \le i \le N$, there is an $m_i$ and
% 		$r_{i,1}, \ldots, r_{i,m_i} \in \dclosure{\bigcup_{i \ge 0}D_i}$
% 		such that
% 		$\setR_i =
% 		\Sigma \setminus \uclosure{\{ r_{i,1}, \ldots, r_{i,m_i}\}}$,
% 	    \item $\setR_i \subseteq \setR_{i+1}$ for $0 \le i < N$,
% 	    \item $N \ge K+1$,
% 	\end{enumerate}
% there is an index $i \le K+1$ such that that $\setR_{i} = \setR_{i+1}$.
% \end{lemma}

\begin{proof}
By downward-finiteness, $\dclosure{D}$ is finite.
Hence, there are only a finite number of different $\setR_i$s we can construct,
and the sequence must converge.
% 
% Recall that $\bigcup_{i\geq 0} D_i$ is finite, say it is the set $\set{d_1,\ldots, d_n}$. 
% 	By Lemma~\ref{prop:finite_di}, $\bigcup_{i \ge 0} D_i = \{ d_1, \ldots, d_n \}$
% 	and $d_1, \ldots, d_n \in \Sigma$.
% 	Hence $r_{i,j} \in \bigcup_{i=1}^n \dclosure{d_i}$ for all
% 	$1 \le i \le N, 1 \le j \le m_i$. The system is downward-finite and all $\dclosure{d_i}$ are finite sets, so
% 	all the $r_{i,j}$ are drawn from a finite set. Thus, there is only a finite number of ways to build sets
% 	$\Sigma \setminus \uclosure{\{r_{i,1},\ldots, r_{i,M_i}\}}$. We denote this number by $K$.
% 
% 	If $N > K$, by condition 1, there are two sets $\setR_{i_1} = \setR_{i_2}$ for some $i_1 < i_2$.
% Using condition 2, it follows that $\setR_{i_1} = \setR_{i_1+1}$.
\qed
\end{proof}

\begin{theorem}{[Termination when $\cover \subseteq \setP$]}
For a given downward-finite WSTS $(\Sigma, I, \to, \wqole)$ and a downward-closed set
$\setP$, if $\cover \subseteq \setP$ and the rule $\rValid$ is applied whenever possible,
% with a higher priority than $\rUnfold$,
then the algorithm terminates and all maximal
	execution sequences have the form $\init \goesto^{*} \resValid$.
\label{thm:termination-uncoverable}
    \end{theorem}

    \begin{proof}
	Consider any execution sequence
	$\init \goesto \sigma_1 \goesto \sigma_2 \goesto \cdots$.
	To show that it is finite, by Proposition~\ref{prop:inner-loop-terminates}, it is sufficient to show
	that there are only finitely many $i$ such that
	$\sigma_i \goesto \sigma_{i+1}$ via rule $\rUnfold$. 
Note that every time $\rUnfold$ is applied, the length of the sequence $\vecR$ goes up.
Consider the bound $K$ obtained by
applying Lemma~\ref{lem:termination} to the finite set $\bigcup_{i\geq 0} D_i$.
After $K$ applications of $\rUnfold$, by Lemma~\ref{lem:termination}, the $\rValid$ rule
applies.
Since $\rValid$ is taken whenever it is applied, the sequence must terminate.
By soundness, it must terminate in $\resValid$. 
% 
% If there are less than $K+1$ such $i$, with $K$ defined as in
% 	Lemma~\ref{lem:termination}, then there are  $R_{i_1} = R_{i_1+1}$, the rule \rValid{} is enabled and the algorithm terminates. 
% 
% Thus, let us assume that there
% are at least $K+1$ distinct values of $i$ such that
% $\sigma_i \goesto \sigma_{i+1}$ via $\rUnfold$. As the length of vector $\vecR$ strictly increases with every application of the $\rUnfold$ rule, there is a $j$ such that for all
% 	$i > j$, the length of $\vecR_i$ is greater than $K$.
% 	By Lemma~\ref{lem:termination},
% 	this implies that for every $i > j$, rule \rValid{} can be applied
% 	to $\sigma_i$. As our assumption is that the rule \rValid{} is applied with a higher priority than $\rUnfold$, 
% this means that for all $i > j$,
% 	rule $\rUnfold$ cannot be applied in $\sigma_i$. Thus, there are at
% 	most $K+1$ values of $i$ such that $\sigma_i \goesto \sigma_{i+1}$ via
% 	$\rUnfold$.
% 
% The proof that a maximal
% execution sequence has the form $\init \goesto^{*} \resValid$ is symmetrical to the proof of Theorem~\ref{thm:termination-reachable}.
\qed
\end{proof}

Note that Theorem~\ref{thm:termination-uncoverable} is the only result that requires
downward-finiteness of the WSTS.
We show that the downward-finiteness condition is necessary.
Consider a WSTS $(\NN \cup \set{\omega}, \set{0}, \to, \leq)$, where $x\to x+1$ for each $x\in\NN$ and $\omega \to\omega$,
and $\leq$ is the natural order on $\NN$ extended with $x \leq \omega$ for all $x\in\NN$.
Consider the downward closed set $\NN$.
The backward analysis terminates in one step, since $\pre{\omega} = \set{\omega}$.
However, the IC3 algorithm need not terminate.
After unfolding, we find a conflict since $\pre{\omega} = \set{\omega}$, which is not initial.
Generalizing, we get $R_1 = \set{0,1}$.
At this point, we unfold again.
We find another conflict, and generalize to $R_2 = \set{0,1, 2}$.
We continue this way to generate an infinite sequence of steps without terminating.

\section{Coverability for Petri Nets}

Petri nets \cite{EsparzaNielsen94} are a widely used model for
concurrent systems.  They form a downward-finite class of WSTS.  We
now describe an implementation of our algorithm for the coverability
problem for Petri nets.

\subsection{Petri Nets}

A Petri net (PN, for short) is a tuple $(S, T, W)$, where $S$ is a
finite set of \emph{places}, $T$ is a finite set of \emph{transitions}
disjoint from $S$, and $W: (S\times T) \cup (T\times S) \rightarrow
\NN$ is the arc multiplicity function.

The semantics of a PN is given using \emph{markings}.  A marking is a
function from $S$ to $\NN$.  For a marking $m$ and place $s\in S$, we
say $s$ has $m(s)$ tokens.

A transition $t \in T$ is \emph{enabled} at marking $m$, written
$m\fire{t}$, if $m(s) \geq W(s,t)$ for all $s\in S$.  A transition $t$
that is enabled at $m$ can fire, yielding a new marking $m'$ such that
$m'(s) = m(s) - W(s,t) + W(t,s)$.  We write $m \fire{t} m'$ to denote
the transition from $m$ to $m'$ on firing $t$.

A PN $(S,T,W)$ and an initial marking $m_0$ give rise to a WSTS
$(\Sigma, \set{m_0}, \to, \wqole)$ as follows.  The set of states
$\Sigma$ is the set of markings.  There is a single initial state
$m_0$.  There is an edge $m \to m'$ if there is some transition $t\in
T$ such that $m\fire{t} m'$.  Finally, $m \wqole m'$ if for each $s\in
S$, we have $m(s) \leq m'(s)$.  It is easy to check that the
compatibility condition holds: if $m_1\fire{t} m_2$ and $m_1 \wqole
m_1'$, then there is a marking $m_2'$ such that $m_1' \fire{t} m_2'$
and $m_2 \wqole m_2'$.  Moreover, the wqo is downward-finite.  The
coverability problem for PNs is defined as the coverability problem on
this WSTS.

We represent Petri nets as follows.
Let $S = \{s_1,\ldots, s_n\}$ be the set of places. 
A marking $m$ is represented as the tuple of natural numbers 
$(m(s_1),\ldots, m(s_n))$. 
% and the set of all markings is a subset of $\NN^n$.  
A transition $t$ is represented as a pair $(\bfg,\bfd)\in\en^n\times\zed^n$,
where $\bfg$ represents the enabling condition, and
$\bfd$ represents the difference between the number of tokens in
a place if the transition fires, and the current number of tokens. 
Formally, $(\bfg,\bfd)$ is defined as:
\begin{align*}
  \bfg & = (W(s_1, t), \ldots, W(s_n, t)) \\
  \bfd & = (W(t, s_1) - W(s_1, t), \ldots, W(t, s_n) - W(s_n, t))\,.
\end{align*}
We represent upward-closed sets with their minimal bases, 
which are finite sets of $n$-tuples of natural numbers. 
A downward-closed set is represented as its complement (which is an upward-closed set). 
The sets $\setR_i$, which are constructed
during the algorithm run, are therefore represented as their
complements. 
Such a representation comes naturally as the algorithm executes. 
Originally each set $\setR_i$ is initialized to contain all the states. 
The algorithm removes sets of states of the form $\uclosure{\bfb}$
from $\setR_i$, for some $\bfb\in \NN^n$. If a set $\uclosure{\bfb}$
was removed from $\setR_i$, we say that states in $\uclosure{\bfb}$
are \emph{blocked} by $\bfb$ at level $i$.
At the end every $\setR_i$ becomes to a
set of the form $\Sigma \setminus \uclosure{\{\bfb_1, \ldots, \bfb_l\}}$ and
we conceptually represent $\setR_i$ with $\{\bfb_1, \ldots, \bfb_l\}$.

The implementation uses a succinct representation of $\vecR$, so called
\emph{delta-encoding} \cite{Een2011}. 
Let $\setR_i = \Sigma \setminus \uclosure{B_i}$ and $\setR_{i+1} = \Sigma
\setminus \uclosure{B_{i+1}}$ for some finite sets $B_i$ and $B_{i+1}$. 
Applying the invariant \eqref{eq:I3} yields $B_{i+1} \subseteq B_i$. 
Therefore we only need to maintain a vector 
$\vecF=(F_0,\ldots,F_N, F_{\infty})$ such that
$\bfb \in F_i$ if $i$ is the highest level where $\bfb$ was blocked. 
This
is sufficient because $\bfb$ is also blocked on all lower levels. As an
illustration, for $(\setR_0, \setR_1, \setR_2)= (\{\bfi_1,\bfi_2\},\{\bfb_1,\bfb_2,\bfb_3,
\bfb_4\}, \{\bfb_2,\bfb_3\})$, the matching vector $\vecF$ might be
$(F_0,F_1, F_2,F_\infty)=(\{\bfi_1,\bfi_2\},\{\bfb_1,\bfb_4\},\{\bfb_2,\bfb_3\},\emptyset)$.
The set $F_{\infty}$ represents states that can never be reached.

\subsection{Implementation Details and Optimizations}

Our implementation follows the rules given in
Figures \ref{fig:rules-generic} and \ref{fig:rules-restricted}. In addition, we
use optimizations from \cite{Een2011}.
The main difference between our implementation and \cite{Een2011}
is in the interpretation of sets being blocked: in \cite{Een2011}
those are cubes identified with partial assignments to boolean variables, whereas in
our case those are upward-closed sets generated by a single state.
Also, a straightforward adaptation of the implementation \cite{Een2011}
would replace a SAT solver with a solver for integer difference logic, a fragment of linear
integer arithmetic which allows the most natural encoding of Petri nets. However, we observed
that Petri nets allow an easy and efficient way of computing predecessors and deciding relative
inductiveness directly. Thus we were able to eliminate the overhead of calling the SMT solver.

\noindent\emph{Testing membership in $\setR_i$.} Many of the rules 
given in Figures~\ref{fig:rules-generic} and \ref{fig:rules-restricted} depend on
    testing whether some state $\bfa$ is contained in a set $\setR_k$. Using
    the delta-encoded vector $\vecF$ this can be done by iterating over $F_i$ for $k\leq i\leq N+1$ and
    checking if any of them contains a state $\bfc$ such that $\bfc\wqole\bfa$. If there is such a state,
    it blocks $\bfa$, otherwise $\bfa\in\setR_k$. If $k=0$, we search for $\bfc$ only in $F_0$.
    
\noindent\emph{Implementation of the rules.} The delta-encoded representation $\vecF$ also makes \rValid{} easy to implement. Checking if
    $\setR_i=\setR_{i+1}$ reduces to checking if $F_i$ is empty for some $i<N$.     
    \rUnfold{} is applied when \rCandidate{} can no longer yield a bad state contained in $\setR_N$.
    It increases $N$ and inserts an empty set to position $N$ in the vector $\vecF$, thus pushing
    $F_\infty$ from position $N$ to $N+1$. We implemented rules \rInitialize{}, \rCandidate{},
    \rModelSyn{} and \rModelSem{} in a straightforward manner.
    
\noindent\emph{Computing predecessors.} In the rest of the rules we need to find predecessors 
$\pre{\uclosure{\bfa}}$ in $\setR_i\setminus\uclosure{\bfa}$,
    or conclude relative inductiveness if no such predecessors exist. The implementation in \cite{Een2011}
    achieves this by using a function \emph{solveRelative()} which invokes the SAT solver.
    But \emph{solveRelative()} also does two important improvements. In case the SAT solver finds a cube of predecessors,
    it applies \emph{ternary simulation} to expand it further. If the SAT solver concludes
    relative inductiveness, it extracts information to conclude a generalized clause is inductive relative to some level
    $k\geq i$. We succeeded to achieve analogous effects in case of Petri nets by the following observations. 
    While it is unclear what ternary simulation would correspond to for Petri nets, the following lemma shows how to
    compute the most general predecessor along a fixed transition directly. 
    
    \begin{lemma}\label{lm:impl-pred}
        Let $\bfa\in\en^n$ be a state and $t=(\bfg,\bfd)\in\en^n\times\zed^n$ be a transition. Then $\bfb\in\pre{\uclosure{\bfa}}$
        is a predecessor along $t$ if and only if $\bfb\wqoge\max(\bfa-\bfd,\bfg)$.
    \end{lemma}
    
    Therefore, to find an element of $\pre{\uclosure{\bfa}}$ and $\setR_i\setminus\uclosure{\bfa}$, we iterate through
    all transitions $t=(\bfg,\bfd)$ and find the one for which $\max(\bfa-\bfd,\bfg)\in\setR_i\setminus\uclosure{\bfa}$.

    If there are no such transitions, then $\Sigma\setminus\uclosure{\bfa}$ is inductive relative to $\setR_i$.
    In that case, for each transition $t=(\bfg,\bfd)$ the predecessor $\max(\bfa-\bfd,\bfg)$ is either
    blocked by $\bfa$ itself, or there is $i_t\geq i$ and a state $\bfc_t\in F_{i_t}$
    such that $\bfc_t\wqole\max(\bfa-\bfd,\bfg)$. We define
      \[ i':=\min\{i_t\ |\ t\text{ is a transition}\}\,, \]
    where $i_t:=N+1$ for $t=(\bfg,\bfd)$ if $\max(\bfa-\bfd,\bfg)$ is blocked by $\bfa$ itself. Then
    $i'\geq i$ and $\Sigma\setminus\uclosure{\bfa}$ is inductive relative to $\setR_{i'}$.
    
\noindent\emph{Computing generalizations.}  The following lemma shows that we can 
also significantly generalize $\bfa$, i.e. there is a simple way
    to compute a state $\bfa'\wqole\bfa$ such that for all transitions $t=(\bfd,\bfg)$, $\max(\bfa'-\bfd,\bfg)$
    remains blocked either by $\bfa'$ itself, or by $\bfc_t$.
    
    \begin{lemma}\label{lem:generalization-c}
        Let $\bfa,\bfc\in\NN^n$ be states and $t=(\bfg,\bfd)\in\NN^n\times\zed^n$ be a transition.
\begin{enumerate}
\item Let $\bfc\wqole\max(\bfa-\bfd,\bfg)$. Define $\bfa''\in\NN^n$ by $a''_j := c_j + d_j$ if $g_j < c_j$
     and $a''_j := 0$ if $g_j \geq c_j$, for $j=1,\ldots, n$.
        % \begin{equation}
        %     \label{eqn:generalization}
        %     a''_j :=
        %     \begin{cases}
        %         c_j + d_j, & g_j<c_j \\
        %         0, & g_j\geq c_j\,,
        %     \end{cases}
        % \end{equation}
        Then $\bfa''\wqole\bfa$. Additionally, for each $\bfa'\in\en^n$ such that
        $\bfa''\wqole\bfa'\wqole\bfa$, we have $\bfc\wqole\max(\bfa'-\bfd,\bfg)$.
\item
        If $\bfa\wqole\max(\bfa-\bfd,\bfg)$, then for each $\bfa'\in\en^n$ such that
        $\bfa'\wqole\bfa$, it holds that $\bfa'\wqole\max(\bfa'-\bfd,\bfg)$.
\end{enumerate}
    \end{lemma}
    
    To continue with the case when the predecessor $\max(\bfa-\bfd,\bfg)$ is blocked
    for each transition $t=(\bfd,\bfg)$, we define $\bfa''_t$ as in Lemma~\ref{lem:generalization-c} (1) if
    the predecessor is blocked by some state $\bfc_t\in F_{i_t}$ and $\bfa''_t:=(0,\ldots,0)$ if it is blocked
    by $\bfa$ itself. 
    The state $\bfa''$ is defined to be the pointwise maximum of all states $\bfa''_t$.    
    By Lemma \ref{lem:generalization-c}, predecessors of $\bfa''$ remain blocked
    by the same states $\bfc_{t}$ or by $\bfa''$ itself.
    
    However, $\bfa''$ still does not have to be a valid generalization, because it might be in $\setR_0$.
    If that is the case, we take any state $\bfc\in F_0$ which blocks $\bfa$ (such a state exists because
    $\bfa\notin\setR_0$). Then $\bfa':=\max(\bfa'',\bfc)$ is a valid generalization: $\bfa'\wqole\bfa$
    and $\Sigma\setminus\uclosure{\bfa'}$ is inductive relative to $\setR_{i'}$.
    
    Using this technique, rules \rDecide{}, \rConflict{} and \rInduction{} become easy to implement. Note that
    some additional handling is needed in rules \rConflict{} and \rInduction{} when blocking a generalized upward-closed
    set $\uclosure{\bfa'}$. If $\Sigma\setminus\uclosure{\bfa'}$ is inductive relative to $\setR_{i'}$ for $i'<N$,
    we update the vector $\vecF$ by adding $\bfa'$ to $F_{i'+1}$. However, if $i'=N$ or $i'=N+1$, we add $\bfa'$
    to $F_{i'}$. Additionaly, for $1\leq k\leq i'+1$ (or $1\leq k\leq i'$) we remove all states $\bfc\in F_k$ such
    that $\bfa'\wqole\bfc$.
    
    One of the optimizations from \cite{Een2011} showed a significant improvement in running time. After
    using the \rConflict{} rule, if $i'+1<N$ and a set $\uclosure{\bfa}$ was blocked from $\setR_{i'+1}$
    by adding a generalization $\bfa'$ to $F_{i'+1}$, we add $\langle\bfa,i'+2\rangle$ to the priority
    queue. This way we do not discard the state which we know leads outside $\setP$, but add an obligation
    to check if its upward-closure can be reached in $i'+2$ steps. The effect is that traces much longer
    than $N$ are checked.

\section{Experimental Evaluation}\label{sec:experiments}

We have implemented the IC3 algorithm in a tool called IIC.
Our tool is written in C++ and uses the input format of mist2.
We evaluated the efficiency of the algorithm on a collection of Petri net examples.
The goal of the evaluation was to compare the performance ---both time and space usage---
of IIC against other implementations of Petri net coverability.
%
% Thus, the goal of the measurement is to show that IIC solves a significant
% amount of problem instances while using less time and space for most of these
% instances, compared to other algorithms.

\begin{table}[t]
    \def\theader#1{\multicolumn{2}{|c|}{#1}}
    \scriptsize
    \begin{tabular}{|l|rr|rr|rr|rr|}
	\hline
	Problem&\theader{IIC}&\theader{Backward}&\theader{EEC}&\theader{MCOV}\\
	Instance&Time&Mem&Time&Mem&Time&Mem&Time&Mem\\
	\hline
	\multicolumn9{|c|}{Uncoverable instances}\\
	\hline
	% Bingham ($h=50$)&$\mathbf{<0.1}$&\bf 6.6&9.4&46.2&0.1&21.3&0.1&20.1\\
	Bingham ($h=150$)&$\bf 0.1$&\bf 3.5&970.3&146.3&1.8&19.0&\bf 0.1&7.6$^{2c}$\\
	Bingham ($h=250$)&\bf 0.2&\bf 6.7&\theader{Timeout}&9.6&45.4&\bf 0.2&19.6$^{2c}$\\
	Ext. ReadWrite (small consts)&\bf 0.0&\bf 1.3&0.1&3.7&\theader{Timeout}&\theader{Timeout/OOM}\\
	Ext. ReadWrite&\bf 0.3&\bf 1.5&216.3&34.1&\theader{Timeout}&0.6&4.1$^{2b}$\\
	FMS (old)&$\mathbf{<0.1}$&\bf 1.3&1.3&5.5&\theader{Timeout}&0.1&5.8$^{2c}$\\
	Mesh2x2&$\mathbf{<0.1}$&\bf 1.3&0.3&3.9&266.9&24.3&$\mathbf{<0.1}$&4.2$^{1c}$\\
	Mesh3x2&$\mathbf{<0.1}$&\bf 1.5&4.1&7.0&\theader{Timeout}&$\mathbf{<0.1}$&2.0$^{2b}$\\
	Multipoll&1.5&\bf 1.6&0.5&4.3&21.8&7.1&$\mathbf{<0.1}$&1.7$^{2b}$\\
	MedAA1&\bf 0.5&\bf 173.3&8.8&598.8&&&3.7&210.4$^{2b}$\\
	MedAA2&\theader{Timeout}&\theader{Timeout}&\theader{}&\theader{Timeout/OOM}\\
	MedAA5&\theader{Timeout}&\theader{Timeout}&\theader{}&\theader{Timeout/OOM}\\
	MedAR1&\bf 0.8&\bf 173.3&8.77&598.8&\theader{}&3.7&210.4$^{2b}$\\
	MedAR2&33.2&\bf 173.3&15.7&599.4&&&\bf 13.7&210.4$^{2b}$\\
	MedAR5&128.1&\bf 173.3&26.6&600&\theader{}&\bf 12.9&210.4$^{2b}$\\
	MedHA1&\bf 0.8&\bf 173.3&8.9&598.7&\theader{}&5.5$^{2c}$&210.4$^{2b}$\\
	MedHA2&33.2&\bf 173.3&14.7&599.5&\theader{}&\bf 12.6&210.4$^{2b}$\\
	MedHA5&\theader{Timeout}&3219.7&647.3&\theader{}&12.5&210.4$^{2b}$\\
	MedHQ1&\bf 0.7&\bf 173.3&8.8&598.8&\theader{}&12.2&210.4$^{2b}$\\
	MedHQ2&33.8&\bf 173.3&16.6&596.9&&&\bf 13.2&210.4$^{2b}$\\
	MedHQ5&125.8&\bf 173.3&26.6&600&\theader{}&\bf 12.6&210.4$^{2b}$\\
	\hline
	\multicolumn9{|c|}{Coverable instances}\\
	\hline
	Kanban&$\mathbf{<0.1}$&\bf 1.4&804.7&55.1&\theader{Timeout}&0.1&6.0$^{2c}$\\
	pncsacover&2.8&\bf 2.2&7.9&11.2&36.5&8.8&\bf 1.0&23.0$^{1c}$\\
	pncsasemiliv&0.1&\bf 1.5&0.2&3.9&32.1&8.8&$\mathbf{<0.1}$&3.7$^{2c}$\\
	MedAA1-bug&\bf 0.8&\bf 172.7&1.0&596.9&56.5&658.0&3.6&210.4$^{2b}$\\
	MedHR2-bug&\bf 0.6&\bf 172.7&0.6&596.9&57.2&658.0&12.8&210.4$^{2b}$\\
	MedHQ2-bug&0.4&\bf 172.7&\bf 0.3&596.9&56.8&658.0&12.9&210.4$^{2b}$\\
	\hline
    \end{tabular}
    \caption{Experimental results: comparison of running time
	and memory consumption for different coverability algorithms on
	selected problem instances.
%The upper half of the table shows
%	problem instances where the coverability property does not hold, while
%	the lower half shows problem instances that are coverable.
	The memory consumption is given in megabytes, and the running
	time in seconds.
	In the mcov column, the superscripts indicate the version of bfc used
	($^1$ means the version Jan 2012 version, $^2$ the Feb 2013 version),
	and the analysis mode ($^c$: combined, $^b$: backward only, $^f$:
	forward only). We list the best result for all the version/parameter
	combinations that were tried.
        % We only show three MedXXX examples as the other examples were similar.
        % For the MCOV
	% column, we used the translations to bfc format provided in the
	% bfc distribution, where available.
    \label{tab:results-selected}
    }
\end{table}
We compare the performance of IIC, using our implementation described above,
to the following algorithms: EEC \cite{EEC} and backward search \cite{ACJT96},
as implemented by the tool mist2\footnote{See \url{http://software.imdea.org/~pierreganty/ist.html}},
and the MCOV algorithm \cite{wahlkroening} for parameterized multithreaded programs
as implemented by bfc\footnote{See \url{http://www.cprover.org/bfc/}}.
All experiments were performed on identical machines, each having
Intel Xeon 2.67 GHz CPUs and 48 GB of memory, running Linux 3.2.21 in 64 bit
mode. % For the mist2 benchmarks, 
Execution time was limited to 1 hours, and memory to five gigabytes.
% , and memory to two gigabytes. For the bfc benchmarks, a shorter time limit of 1.5 hours was imposed. 
% CPU and memory consumption was measured using the time
% command, accounting for user-space CPU time and wall-clock time, and
% maximal resident set size.

We used 29 Petri net examples from the mist2 distribution,
46 examples of multi-threaded programs from the bfc distribution,
and 
6 examples from checking security properties of 
message-passing programs communicating through unbounded unordered channels (MedXXX examples). 
% Additionally, we used 12 Petri nets provided by our colleague Zilong Wang.
%
We only present a selection of the data and focus on examples
that took longer than 2 second for at least one algorithm.
% We also elided some of the MedXXX examples, because the measurements were all very similar to each other.
All benchmarks are available at \url{http://www.mpi-sws.org/~jkloos/iic-experiments}.

\begin{table}[t]
    \def\theader#1{\multicolumn{2}{|c|}{#1}}
    \scriptsize
    \begin{tabular}{|l|rr|rr|}
	\hline
	Problem&\theader{IIC}&\theader{MCOV}\\
	Instance&Time&Mem&Time&Mem\\
	\hline
	\multicolumn{5}{|c|}{Uncoverable instances}\\
	\hline
	Conditionals 2&0.1&\bf 3.6&$\mathbf{<0.1}$&5.7$^{2c}$\\
	RandCAS 2&$\mathbf{<0.1}$&\bf 2.0&$\mathbf{<0.1}$&3.9$^{2c}$\\
	\hline
	\multicolumn{5}{|c|}{Coverable instances}\\
	\hline
%	Boop 1&$<0.1$&\bf 5.5&$<0.1$&13.7\\
	Boop 2&82.0&287.9&\bf 0.1&\bf 12.1$^{1c}$\\
	%Boop 2&82.0&1151.6&\bf 0.1&\bf 49.2\\
	FuncPtr3 1&$<0.1$&\bf 1.5&$<0.1$&3.4$^{2c}$\\
	FuncPtr3 2&0.2&\bf 12.3&0.1&7.9$^{2c}$\\
	FuncPtr3 3&28.5&939.1&\bf 3.6&\bf 303.8$^{1c}$\\
	%FuncPtr3 3&28.5&3756.2&\bf 9.7&\bf 1669.4$^{2c}$\\
	DoubleLock1 2&\theader{Timeout}&\bf 0.8&\bf 56.7$^{2c}$\\
	DoubleLock3 2&8.0&41.3&$\mathbf{<0.1}$&\bf 4.8$^{2c}$\\
	Lu-fig2 3&\theader{Timeout}&\bf 0.1&\bf 10.4$^{2c}$\\
	Peterson 2&\theader{Timeout}&\bf 0.2&\bf 23.0$^{1c}$\\
	Pthread5 3&132428&468.8&\bf 0.1&\bf 17.0$^{1c}$\\
	Pthread5 3&&&0.2&\bf 49.6$^{2c}$\\
	SimpleLoop 2&7.9&6.0&$\mathbf{<0.1}$&\bf 4.8$^{2c}$\\
	Spin2003 2&4852.2&54.4&$\mathbf{<0.1}$&\bf 2.7$^{2c}$\\
	StackCAS 2&2.5&\bf 1.6&$\mathbf{<0.1}$&3.7$^{2c}$\\
	StackCAS 3&5.5&21.7&$\mathbf{<0.1}$&\bf 4.4$^{2c}$\\
	Szymanski 2&\theader{Timeout}&\bf 0.4&\bf 26.7$^{2c}$\\
	\hline
    \end{tabular}
    \caption{Experimental results: comparison between
	MCOV and IIC on examples derived from
        parameterized multithreaded programs.
	In the mcov column, the superscripts indicate the version of bfc used
	($^1$ means the version Jan 2012 version, $^2$ the Feb 2013 version),
	and the analysis mode ($^c$: combined, $^b$: backward only, $^f$:
	forward only). We list the best result for all the version/parameter
	combinations that were tried.
    \label{tab:results-wk}}
\end{table}

\smallskip
\noindent
\emph{mist2 and MedXXX benchmarks}
Table~\ref{tab:results-selected} show run times and memory usage on the mist2 and message-passing
program benchmarks.
For each row, the column in bold shows the winner (time or space) for each instance.
It can be seen that IIC performs reasonably well on these benchmarks, both in time and in memory usage.
% \JK{Possible addition: Some further experiments show that the good 
% performance of the backward algorithm on the MedXXX examples stems from their use 
% of pre-compouted invariants. We conjecture that making using of the invariants in the IIC implementation would provide similar benefits.}

To account for mist2's use of a pooled memory, we estimated its baseline usage to 2.5 MB
by averaging over all examples that ran in less than 1 second.
% we checked whether mist2 uses significantly more
% memory than baseline usage (estimated by finding the average for examples
% that executed in less than one second). We found that baseline usage was around
% 10 megabytes, and the listed examples all required at least 5 megabytes more,
% rendering this measurement significant.

% Also, the memory usage includes the size of the program binary.
% To compensate, we created statically-linked versions of all
% binaries and checked the memory requirements of the code for each.
% It turns out that MCOV takes roughly 2 megabytes, IIC requires 1.2 megabytes and mist2
% 750 kilobytes.
% Again, this does not impact our measurements, since the differences
% between the memory requirements tend to be more significant.

\noindent
\emph{Multithreaded program benchmarks}
We also ran comparisons with MCOV on a set of multithreaded programs distributed with MCOV.
For Petri nets derived from C programs distributed with MCOV, Table~\ref{tab:results-wk}
shows that IIC performs well on the uncoverable examples but MCOV performs much
better on the coverable ones.
We do not fully understand the reasons for poor performance of IIC for the coverable instances.
% Part of the poor performance can be explained by the compilation process used to generate a Petri net
% from the MCOV input format.
% While MCOV can take advantage of the structure of the program (e.g., by distinguishing local and shared
% states, where shared states are 1-bounded),
% the compiled Petri net that is fed to IIC does not have this information and the algorithm
% explores many infeasible backward reachable states.

%In conclusion, IIC performs favorably on a large set of benchmarks. It will be interesting to see if 
%the generalization heuristic from \cite{wahlkroening} can be combined with IIC.
In conclusion, the unoptimized implementation of the IIC algorithm is already
working quite well in comparison to other existing implementations of
coverability algorithms. Nevertheless, it is obvious that significant further
work is required to optimize the algorithm. Two main directions that are being
considered are the use of invariants to prune the search space, and the
combination of the generalization heuristics from MCOV \cite{wahlkroening} with
IIC.

\smallskip
\noindent\textbf{Acknowledgements}
We thank Andreas Kaiser for pointing out an error regarding the encoding
of Petri nets into the bfc input format, leading to non-optimal performance
of the bfc tool, and for providing us with a correct conversion tool.

    \bibliographystyle{plain}
    \bibliography{references}

\begin{thebibliography}{10}

\bibitem{ACJT96}
P.~A. Abdulla, K.~Cerans, B.~Jonsson, and Yih-Kuen Tsay.
\newblock General decidability theorems for infinite-state systems.
\newblock In {\em LICS '96}, pages 313--321. IEEE, 1996.

\bibitem{ABJ98}
P.A. Abdulla, A.~Bouajjani, and B.~Jonsson.
\newblock On-the-fly analysis of systems with unbounded, lossy {FIFO} channels.
\newblock In {\em CAV$\:$'98}, LNCS 1427, pages 305--318. Springer, 1998.

\bibitem{Bradley2011}
A.R. Bradley.
\newblock {SAT}-based model checking without unrolling.
\newblock In {\em VMCAI'11}, LNCS, pages 70--87. Springer, 2011.

\bibitem{Cia94}
G.~Ciardo.
\newblock Petri nets with marking-dependent arc multiplicity: properties and
  analysis.
\newblock In {\em ICATPN$\:$'94}, volume 815 of {\em LNCS}, pages 179--198.
  Springer, 1994.

\bibitem{Cimatti}
A.~Cimatti and A.~Griggio.
\newblock Software model checking via {IC3}.
\newblock In {\em CAV'12: Computer-Aided Verification}, LNCS 7358, pages
  277--293. Springer, 2012.

\bibitem{Dickson1913}
L.E. Dickson.
\newblock Finiteness of the odd perfect and primitive abundant numbers with n
  distinct prime factors.
\newblock {\em American Journal of Mathematics}, 35(4):413--422, 1913.

\bibitem{DFS98}
C.~Dufourd, A.~Finkel, and P.~Schnoebelen.
\newblock Reset nets between decidability and undecidability.
\newblock In {\em ICALP$\:$'98}, LNCS 1443, pages 103--115. Springer, 1998.

\bibitem{Een2011}
N.~Een, A.~Mishchenko, and R.~Brayton.
\newblock Efficient implementation of property directed reachability.
\newblock In {\em FMCAD'11}, pages 125--134. FMCAD Inc, 2011.

\bibitem{EN98}
E.A. Emerson and K.S. Namjoshi.
\newblock On model checking for non-deterministic infinite-state systems.
\newblock In {\em LICS$\:$'98}, pages 70--80. IEEE, 1998.

\bibitem{EFM99}
J.~Esparza, A.~Finkel, and R.~Mayr.
\newblock On the verification of broadcast protocols.
\newblock In {\em LICS$\:$'99}, pages 352--359. IEEE Computer Society, 1999.

\bibitem{EsparzaNielsen94}
J.~Esparza and M.~Nielsen.
\newblock Decidability issues for petri nets - a survey.
\newblock {\em Bulletin of the EATCS}, 52:244--262, 1994.

\bibitem{FS01}
A.~Finkel and P.~Schnoebelen.
\newblock Well-structured transition systems everywhere!
\newblock {\em Theor. Comput. Sci.}, 256(1-2):63--92, 2001.

\bibitem{EEC}
G.~Geeraerts, J.-F. Raskin, and L.~Van~Begin.
\newblock Expand, enlarge and check: New algorithms for the coverability
  problem of {WSTS}.
\newblock {\em J. Comput. Syst. Sci.}, 72(1):180--203, February 2006.

\bibitem{Higman1952}
G.~Higman.
\newblock Ordering by divisibility in abstract algebras.
\newblock {\em Proceedings of the London Mathematical Society},
  s3-2(1):326--336, 1952.

\bibitem{Hoder2012}
K.~Hoder and N.~Bj{\o}rner.
\newblock Generalized property directed reachability.
\newblock In {\em SAT'12}, pages 157--171. Springer, 2012.

\bibitem{wahlkroening}
A.~Kaiser, D.~Kroening, and T.~Wahl.
\newblock Efficient coverability analysis by proof minimization.
\newblock {\em CONCUR 2012--Concurrency Theory}, pages 500--515, 2012.

\end{thebibliography}

\appendix
\section{Soundness and termination proof}
\label{ap:proofs}

This appendix contains the proofs of lemmas used in the paper.

\begin{theorem-non}\textbf{\emph{\ref{lemma:rulesI14invs}}}
The rules
\rUnfold, \rInduction, \rConflict, \rCandidateUnres{}, and
\rDecideUnres{} preserve \eqref{eq:I1} -- \eqref{eq:I4},
\end{theorem-non}

\begin{proof}
If $\vecR|Q \goesto \vecRI|Q'$ by application of
\rCandidateUnres{} or \rDecideUnres, we have $\vecR = \vecRI$,
so \eqref{eq:I1} -- \eqref{eq:I4} are preserved trivially.
For \rUnfold, \eqref{eq:I1} -- \eqref{eq:I3} are trivial, and
\eqref{eq:I4} holds for $i < \len{\vecRI} - 1$ by \eqref{eq:I4} on
$\vecR$, and by the condition of \rUnfold{} for
$i = \len{\vecRI} - 1$.

Finally, the rules \rInduction{} and \rConflict{} require the following technical observation about $\operatorname{Gen}$.
\begin{verse}
\textbf{Claim:}
If $\vecR$ satisfies \eqref{eq:I1} -- \eqref{eq:I4} and
$b \in \operatorname{Gen}_m(a)$, then
$\vecR[\setR_i \gets \setR_i \setminus \uclosure{b}]_{i=1}^m$
satisfies \eqref{eq:I1} -- \eqref{eq:I4}.
\end{verse}
To prove the claims, we show the following:
	\begin{enumerate}
	    \item For $1 \le k \le i$,
		$I \subseteq \setR_k \setminus \uclosure{b}$
		(part of \eqref{eq:I1}).
	    \item $\post{\setR_{i-1} \setminus \uclosure{b}} \subseteq
		\setR_i \setminus \uclosure{a'}$.
		for $1 < i < m$ (Part of \eqref{eq:I2}).
	    \item $\post{\setR_0} \subseteq \setR_1 \setminus \uclosure{b}$.
		(\eqref{eq:I2}, case $i = 1$)
	    \item $\post{\setR_{m-1} \setminus \uclosure{b}} \subseteq \setR_m$
		(\eqref{eq:I2}, case $i = m$)
	\end{enumerate}
All other cases as well as \eqref{eq:I3} and \eqref{eq:I4} are trivial.

\begin{enumerate}
	    \item By the definition of $\operatorname{Gen}$, we have $\uclosure{b} \cap I =
		\varnothing$. Thus, since $I \subseteq \setR_i$ by
		\eqref{eq:I1}, $I \subseteq \setR_i \setminus \uclosure{b}$.
	    \item Let $i$ be given with $1 < i < m$, and
		$y \in \post{\setR_{i-1} \setminus \uclosure{b}}$.
		We need to show that $y \in \setR_i \setminus \uclosure{b}$.

		By choice of $y$, there is an
		$x \in \setR_{i-1} \setminus \uclosure{b}$ such that
		$x \to y$. By repeated application of \eqref{eq:I3}, we find
		that $x \in \setR_{m-1} \setminus \uclosure{b}$. Thus,
		$y \in \post{\setR_{m-1} \setminus \uclosure{b}}
		 \subseteq \Sigma \setminus \uclosure{b}$.

		Thus, $y \in \setR_i \setminus \uclosure{b}$.
	    \item Let $y' \in \post{\setR_0} = \post{\dclosure{I}}$.
		We need to show that $y' \in \setR_1 \setminus \uclosure{b}$.

		There is a $x' \in \dclosure{I}$ such that $x' \to y'$.
		Due to the choice of $x'$, there is an $x \in I$ with
		$x \wqoge x'$. By well-structuredness, there is also a $y$
		such that $x \to y$ and $y \wqoge y'$. Since $\setR_1$ is
		downward-closed, $y \in \setR_1$.

		By \eqref{eq:I3}, we find that
		$x \in \setR_1$, and by (1),
		$x \in \setR_1 \setminus \uclosure{b}$.
		Thus, by (2), $y \in \setR_2 \setminus \uclosure{b}$.
		But this implies $y \not\in \uclosure{b}$, so
		$y \in \setR_1 \setminus \uclosure{b}$. Since
		$\setR_1 \setminus \uclosure{b}$ is downward-closed,
		we hence have $y' \in \setR_1 \setminus \uclosure{b}$.
	    \item $\post{\setR_{m-1} \setminus \uclosure{b}}
		\subseteq \post{\setR_{m-1}} \subseteq \setR_m$ by
		\eqref{eq:I2}.
	\end{enumerate}
\qed
\end{proof}

The next lemma defines the structure of the priority queues
used in the algorithm.

\begin{theorem-non}\textbf{\emph{\ref{lem:Q-structure}}}
	Let $\init \goesto^{*} \vecR\mid Q$. 
  If $Q \neq \varnothing$, then for every $\langle a, i\rangle \in Q$,
    there is a path from $a$ to some $b \in \Sigma \setminus \setP$.
\end{theorem-non}

    \begin{proof}
    By induction on the application of rules. For the base case,
    the application of \rInitialize, the claim trivially holds.
    
    For the induction step, assume the claim holds for some sequence
    of rule applications such that $\init \goesto^{*} \vecR\mid Q$.
    We only need to consider \rCandidateUnres{} and \rDecideUnres{},
    since they are the only rules which add elements on $Q$.
    
    If \rCandidateUnres{} is applied, it will enqueue $\langle a,N\rangle$
    such that $a\in \setR_N\setminus\setP\subseteq\Sigma\setminus\setP$.
    If \rDecideUnres{} is applied, then $\min Q = \langle a,i \rangle$,
    $i>0$ and $\langle b',i-1\rangle$ such that $b'\in\pre{\uclosure{a}}$
    is enqueued. The latter implies there is $a'\wqoge a$ such that
    $b'\to a'$. By the induction hypothesis, there is a path
    $a\to^{*}b\in\Sigma\setminus\setP$, therefore by
    well-structuredness there is $b''\wqoge b$ such that $a'\to^{*}b''$. 
    Combining the facts we conclude $b'\to^{*}b''\in\Sigma\setminus\setP$.
    \qed
    \end{proof}

%    \begin{remark}[Various minor invariants]
%	\label{APP-rem:minor-invs}
%	The following invariants of the algorithm are easy to verify:
%	\begin{enumerate}
%	    \item If $\init \goesto \sigma_1 \goesto^{*} \sigma_n$,
%		we have that $\sigma_i$ is of the form $\vecR_i|Q_i$ for all
%		$1 \le i < n$,
%		and $\sigma_n$ is of one of three forms: $\vecR_n|Q_n$,
%		\resValid{} or \resInvalid.
%	    \item As a corollary to (1), if
%		$\init \goesto \sigma_1 \goesto \sigma_2 \goesto \cdots$
%		with an infinite sequence of $\sigma_i$, all $\sigma_i$ are
%		of the form $\vecR_i|Q_i$.
%
%		Conversely, if $\sigma_i = \resValid$ or
%		$\sigma_i = \resInvalid$, this is the last element of the
%		sequence, i.e., $n = i$.
%	    \item $\init \not\goesto \resValid$ and
%		$\init \not\goesto \resInvalid$.
%	    \item For every sequence
%		$\init \goesto \vecR_1|Q_1 \goesto^{*} \vecR_n|Q_n$,
%		$\operatorname{length}(\vecR_i) \le
%		\operatorname{length}(\vecR_{i+1})$. In particular,
%		when $\vecR_i|Q_i \goesto \vecR_{i+1}|Q_{i+1}$ via
%		\rUnfold, $\operatorname{length}(\vecR_i) <
%		\operatorname{length}(\vecR_{i+1})$.
%	    \item As a corollary to (4), when
%		$\init \goesto \vecR_1|Q_1 \goesto^{*} \vecR_n|Q_n$,
%		and there are $K$ different indices $i$ such that
%		$\vecR_i|Q_i \goesto \vecR_{i+1}|Q_{i+1}$ via the rule
%		\rUnfold, there is an index $j \le n$ such that for all
%		$i \ge j$, $\operatorname{length}(\vecR_i) \ge K$.
%	\end{enumerate}
%    \end{remark}

    \begin{lemma}[Disjointness of $\setR_i$ and $U_j$]
	\label{APP-lem:R_U_disjoint}
	When \eqref{eq:I1} -- \eqref{eq:I4} hold for $\vecR$,
	$R_{N-1-i} \cap U_i = \varnothing$ for $0 \le i < N - 1$.
    \end{lemma}

    \begin{proof}
	We prove the statement by induction over $i$.

	\begin{description}
	    \item[$i = 0$:] By \eqref{eq:I4},
	\[ R_{N-1-0} \cap U_0 = \underbrace{R_{N-1}}_{\subseteq \setP} \cap
	    (\Sigma \setminus \setP) = \varnothing. \]
	
	    \item[$i > 0$:] By induction, $R_{N-i} \cap U_{i-1} = \varnothing$.
		Now, let $x \in U_i$. Then by definition of $U_i$, there are
		two cases:
		\begin{description}
		    \item[$x \in U_{i-1}$:] Then $x \not\in \setR_{N-i}$.
			Since $\setR_{N-i-1} \subseteq R_{N-i}$ by
			\eqref{eq:I3}, $x \not\in \setR_{N-i-1}$.
		    \item[$x \in \pre{U_{i-1}}$:] Then there is a
			$y \in U_{i-1}$	such that $x \to y$. In particular,
			$y \in \post{x}$. Since	$y \in U_{i-1}$, we also have
			$y \not\in \setR_{N-i}$. By \eqref{eq:I2} and
			$z \in \setR_{N-i-1} \implies \post{z} \subseteq
			\post{\setR_{N-i-1}}$, this implies
			$x \not\in \setR_{N-i-1}$.
		\end{description}
		Thus, in either case, $x \not\in \setR_{N-i-1}$. This implies
		$\setR_{N-i-1} \cap U_i = \varnothing$.
	\end{description}
\qed
    \end{proof}

 \begin{lemma}\label{APP-lem:D_iProperties}
The sets $D_i$ satisfy the following properties:
\begin{enumerate}
  \item $D_0 \subseteq \Sigma \setminus \setP$
  \item  $D_{i+1} \subseteq \pre{D_i} \setminus U_i$
  \item Whenever $\setR_N \setminus \setP \neq \varnothing$, there exists an $x \in \setR_N \cap D_0$
  \item For all $a \in D_i$, if $\pre{\uclosure{a}} \cap \setR_{N-i-1} \neq \varnothing$, there exists an element $x$ such that
$x \in \pre{\uclosure{a}} \cap D_{i+1} \cap \setR_{N-i-1}$
  \item $D_i$ is finite for all $i \ge 0$
\end{enumerate}
\end{lemma}

    \begin{proof}
	Statements 1) and 2) follow trivially.

	To prove (3), assume that $y \in \setR_N \setminus \setP$. Then
	there is a minimal element $x \in \min(\setR_N \setminus \setP)$.
	But since $\setR_N$ is downward-closed,
	$\min (\setR_N \setminus \setP)
	\subseteq \min (\Sigma \setminus \setP)$.
	Thus, $x \in \min(\Sigma \setminus \setP) = D_0$. $x \in \setR_N$ is
	clear.

	To show (4), let $a \in D_i$ be given, and assume that
	$y \in \pre{\uclosure{a}} \cap \setR_{N-i-1}$. Again, there is a
	minimal element $x \in \pre{\uclosure{a}} \cap \setR_{N-i-1}$.
	By Lemma \ref{APP-lem:R_U_disjoint}, $x \not\in U_i$. Thus, $x \in D_{i+1}$.

	Finally, (5) follows by induction on $i$: For $i = 0$, the
	statement is clear because of the finiteness of $\min$. For $i > 0$,
	the set $D_{i-1}$ is finite by induction hypothesis. Thus,
	the union $\bigcup_{a \in D_{i-1}} \min(\uclosure{a})$ is a finite
	union over finite sets, thus $D_i$ is a subset of a finite set and
	hence finite.
\qed
    \end{proof}

    \begin{lemma}
	\label{APP-lem:finite-choice}
Given a WSTS $(\Sigma, I, \to, \wqole)$, a downward-closed set
$\setP$ and a sequence of sets $D_i$, if $\init \goesto^{*} \vecR|Q$, then:
	\begin{enumerate}
	    \item For all $i \ge 1$,
		$\setR_i = \Sigma \setminus \{ r_{i,1}, \ldots, r_{i,m_i} \}$,
		where for all $j = 1, \ldots, m_i$, there is a
		$k \ge 0$ and a $d \in D_k$ such that $r_{i,j} \le d$.
	    \item For all $\langle a, i\rangle \in Q$, $a \in D_{N-i}$.
	\end{enumerate}
    \end{lemma}

    \begin{proof}
	It is again sufficient to show that $\dclosure{I}|\varnothing$
	has this property, and that all relevant rules preserve it.
	Since $\dclosure{I}|\varnothing$ satisfies the requirements vacuously,
	assume that $\vecR|Q \goesto \vecRI|Q'$. By inspection, the following
	five rules need to be considered:
	\begin{description}
	    \item[\rUnfold] Trivial.
	    \item[\rInduction] Since $Q = Q'$, the second part is trivial.

		For the first part, let $b \in \gen_i(r_{i,j})$ for given
		$i,j$. By the definition of $\gen$, $b \le r_{i,j}$, and
		by induction hypothesis, $r_{i,j} \le d$ for some $d \in D_k$,
		$k > 0$. By transitivity, $b \le d$.
		Furthermore, \begin{equation*}
\begin{split}
		    \setRI_\ell & = \begin{cases}
			\Sigma \setminus
			  \uclosure{\{ r_{\ell,1}, \ldots, r_{\ell,m_\ell}\}}
			  \setminus \uclosure{b} & 1 \le \ell \le i+1\\
			\Sigma \setminus
			  \uclosure{\{ r_{\ell,1}, \ldots, r_{\ell,m_\ell}\}}
			  & \text{otherwise}
		      \end{cases} \\
           & = \begin{cases}
			\Sigma \setminus
			  \uclosure{\{ r_{\ell,1}, \ldots, r_{\ell,m_\ell}, b\}}
			  & 1 \le \ell \le i+1\\
			\Sigma \setminus
			  \uclosure{\{ r_{\ell,1}, \ldots, r_{\ell,m_\ell}\}}
			  & \text{otherwise}
		      \end{cases}
		\end{split}
\end{equation*}

		So, in either case, $\setRI_\ell$ is of the required form.
	    \item[\rCandidate] Trivial.
	    \item[\rDecide] Trivial.
	    \item[\rConflict] Since $Q' \subseteq Q$, the second part is
		trivial. For the first part, we have $b \in \gen_i(a)$
		for some $a \in D_k$, $k \ge 0$. Thus, $b \le a$ by the
		definition of $\gen$. The rest of the proof is analogous to
		the case of \rInduction.
	\end{description}
\qed
    \end{proof}

    \begin{lemma}[Progress]
	\label{APP-lem:progress}
	If $\init \goesto^{*} \vecR|Q$, then either $\vecR|Q \goesto \vecRI|Q'$, or $\vecR|Q \goesto \resValid$, or
	$\vecR|Q \goesto \resInvalid$.
    \end{lemma}

    \begin{proof}
	Let $\vecR | Q$ be given. By case analysis, we will show that
	some rule will always be applicable to it.

	If $Q = \varnothing$, there are two cases:
	\begin{itemize}
	    \item $\setR_N \subseteq \setP$.

		$\vecR|Q = \vecR|\varnothing \goesto \vecR \cdot
		\Sigma|\varnothing$
		by applying \rUnfold.
	    \item $\setR_N \not\subseteq \setP$.
		
		Then, by choice of $D_0$, there is some $x \in \setR_N \cap D_0$.
		Thus, $\vecR|Q =\vecR|\varnothing \goesto
		\vecR|\langle x,N\rangle$
		by applying \rCandidate.
	\end{itemize}

	If $Q \neq \varnothing$ is not empty, there are four cases:
	\begin{itemize}
	    \item $\langle a, 0 \rangle \in Q$ for some $a \in \Sigma$.

		$\vecR|Q \goesto \resInvalid$ by applying \rModelSyn.
	    \item $\langle a, i \rangle \in Q$ for some $a \in \Sigma$,
		$i \ge 0$ with $\uclosure{a} \cap I \neq \varnothing$.

		$\vecR|Q \goesto \resInvalid$ by applying \rModelSem.
	    \item $\min Q = \langle a, i\rangle$ for some $a \in \Sigma$,
		$i > 0$ with
		$\pre{\uclosure{a}} \cap \setR_{i-1} \neq \varnothing$.

		By choice of $D_{N-i+1}$, there is also a
		$b \in D_{N-i+1} \cap \setR_{i-1} \cap \pre{\uclosure{a}}$, so
		$\vecR|Q \goesto \vecR|\push{Q}{\langle b, i-1\rangle}$
		by applying \rDecide.
	    \item None of the above.

		In this case, let $\langle a, i\rangle = \min Q$.
		We have $i > 0$, $\uclosure{a} \cap I = \varnothing$ and
		$\pre{\uclosure{a}} \cap \setR_{i-1} = \varnothing$.

		\begin{description}
		    \item[Claim:] $a \in \operatorname{Gen}_i(a)$.
		    \item[Proof:] We certainly have $a \le a$, and by the
			statements above, $\uclosure{a} \cap I = \varnothing$.	
			Also, by Lemma \ref{lem:Q-structure}, $a \in \setR_i$.
			It remains to show that
			$\post{\setR_{i-1} \setminus \uclosure{a}}
			 \subseteq \Sigma \setminus \uclosure{a}$.

			Thus, let
			$y \in \post{\setR_{i-1} \setminus \uclosure{a}}$.
			Then there is an
			$x \in \setR_{i-1} \setminus \uclosure{a}$
			such that $x \to y$.
			
			Suppose now that
			$y \in \uclosure{a}$. Then $x \in \pre{\uclosure{a}}$,
			so $x \in \pre{\uclosure{a}} \cap \setR_{i-1}
			= \varnothing$ -- contradiction.
			
			Thus, $a \in \operatorname{Gen}_i(a)$.
		\end{description}

		Thus, $\vecR|Q \goesto
		 (\vecR[\setR_k \gets \setR_k \setminus \uclosure{a}]_{k=1}^i)|
		 (\popMin{Q})$ by applying \rConflict.
	\end{itemize}
\qed
    \end{proof}

\begin{prop-non}\textbf{\emph{\ref{cor:max-finite-seqs}}} [Maximal finite sequences]
	Let $\init = \sigma_0 \goesto \sigma_1 \goesto \cdots \goesto \sigma_K$
	be a maximal sequence of states, i.e., a sequence such that there is
	no $\sigma'$ such that $\sigma_K \goesto \sigma'$. Then
	$\sigma_K = \resValid$ or $\sigma_K = \resInvalid$.
    \end{prop-non}

    \begin{proof}
	$\sigma_K$ can have four values, $\init$, $\resValid$, $\resInvalid$
	or $\vecR|Q$.

	If $\sigma_K = \resValid$ or $\sigma_K = \resInvalid$. %  we are done by Remark \ref{rem:minor-invs}.

	If $\sigma_K = \vecR|Q$, the sequence is not maximal by Lemma
	\ref{APP-lem:progress}.

	If $\sigma_K = \init$, $\sigma_K \goesto \dclosure{I}|\varnothing$,
	by \rInitialize, hence the sequence is not maximal.
\qed
    \end{proof}

\begin{theorem-non}\textbf{\emph{\ref{lem:les_wfqo}}} [$\le_s$ is a well-founded quasi-order.]
	\label{APP-lem:les_wfqo}
	The relation $<_s$ is a well-founded strict quasi-ordering on the set
	$(\mathcal D)^{*} \times \mathcal Q$, where
	$\mathcal D$ is a set of downward-closed sets over $\Sigma$, and
	$\mathcal Q$ denotes the set of priority queues over
	$\Sigma \times \en$.
    \end{theorem-non}

    \begin{proof}
	The following statements are easy to check:
	\begin{itemize}
	    \item $\sqsubseteq_N$ is a partial order, and $\sqsubset_N$ is
		its strict part.
	    \item $\sqsubseteq$ is a partial order, and $\sqsubset$ is its
		strict part.
	    \item Let $\le_n := \sqsubseteq \times_{\operatorname{lex}} \le$
		denote the lexicographical product of $\sqsubseteq$ and
		the order $\le$ on the natural numbers. Then $\le_n$ is
		a partial order.
	    \item Let
		$\phi: (\mathcal D)^{*} \times \mathcal Q \to
		(\mathcal D)^{*} \times \en,
		\vecR|Q \mapsto (\vecR, \ell_{\len(\vecR)}(Q))$.

		Then $\phi(\vecR|Q) <_n \phi(\vecRI|Q')$ if
		$\vecR|Q <_s \vecRI|Q'$, and
		$\phi(\vecR|Q) \le_n \phi(\vecRI|Q')$ if
		$\vecR|Q \le_s \vecRI|Q'$.
	    \item If $\le_s$ is a quasi-order, $<_s$ is the corresponding
		strict quasi-order.
	\end{itemize}

	In the following, we will use these facts to establish:
	\begin{enumerate}
	    \item $\sqsubseteq$ is well-founded,
	    \item $\le_n$ is well-founded,
	    \item $\le_s$ is a quasi-order,
	    \item $\le_s$ is well-founded.
	\end{enumerate}

	\begin{description}
	    \item[$\sqsubseteq$ is well-founded:]
		Let $\vecR_1 \sqsupseteq \vecR_2 \sqsupseteq \cdots$ be
		a descending chain of vectors. We need to show that the chain
		will eventually stabilize, i.e., there is an $i$ such that
		for all $j \ge i$, $\vecR_j = \vecR_i$.

		As a first observation, by definition of $\sqsubseteq$,
		$\len \vecR_j = \len \vecR_{j+1}$ for all $j \ge 0$, i.e.,
		there is an $N$ such that $\len \vecR_j = N$ for all $j$.

		Suppose that no such $i$ exists. Then for all $j$,
		$\vecR_{j+1} \sqsupset_N \vecR_j$. By definition of
		$\sqsupset_N$, this means that for every $j$, there is a
		$k_j$ such that $\setR_{j,k_j} \supsetneq \setR_{j+1,k_j}$.

		Furthermore, since $k_j \in \{1, \ldots N\}$ for all $j$,
		there must be some $k \in \{ 1, \ldots, N \}$ such that
		$k_j = k$ for infinitely many $j$ by the pigeonhole principle.

		Define a sequence $j_t$ such that $j_0 = 0$, and for all
		$t \ge 0$,
		$\setR_{j_t,k} = \setR_{j_{t+1}-1,k} \supsetneq \setR_{j_t,k}$.
		Such a sequence exists because for every $j$, either
		$\setR_{j,k} = \setR_{j+1,k}$, or
		$\setR_{j,k} \supsetneq \setR_{j+1,k}$ by the assumptions.

		Thus, we have an infinite descending chain
		$\setR_{j_0,k} \supsetneq \setR_{j_1,k} \supsetneq \ldots$
		of downward-closed sets. Define
		$\uclose{C}_t := \Sigma \setminus \setR_{j_t,k}$. This is
		an infinite strictly ascending chain of upward-closed sets,
		i.e., $\uclose{C}_0 \subsetneq \uclose{C}_1 \subsetneq \ldots$.
		This is a contradiction, since there are now infinite
		strictly ascending chains of upward-closed sets, cf.
		\cite{ACJT96}, Lemma 3.4.
	    \item[$\le_n$ is well-founded:]
		Assume that $s$ is an infinite descending sequence on
		$\mathcal{D}^{*} \times \en$. Denote by $s_1$ the sequence
		of first components and by $s_2$ the sequence of second
		components, i.e., $s(i) = (s_1(i), s_2(i))$.
		Since $\sqsubseteq$ is a
		well-founded partial order, there is some $j$ such that
		$s(k) = s(j)$ for all $k \ge j$. Thus, for $s(k) > s(\ell)$
		for all $j \le k < \ell$, which is impossible, since $\le$
		is well-founded.
	    \item[$\le_s$ is a quasi-order:]
		Reflexivity is trivial. Consider
		$\vecR_1|Q_1 \le_s \vecR_2|Q_2 \le_s \vecR_3|Q_3$.

		By definition, $\vecR_1 \sqsubseteq \vecR_2 \sqsubset \vecR_3$,
		hence $\vecR_1 \sqsubseteq \vecR_3$. Additionally, due to the
		definition of $\sqsubseteq$, there is an $N$ such that
		$N = \len \vecR_1 = \len \vecR_2 = \len \vecR_3$.

		There are three cases to consider:
		\begin{enumerate}
		    \item $\vecR_1 = \vecR_2 = \vecR_3$. In this case,
			$\ell_{\len(R_1)}(Q_1) \le \ell_{\len(R_2)}(Q_2)
			 \le \ell_{\len(R_3)}(Q_3)$.

			By the above observation, this means that
			$\ell_N(Q_1) \le \ell_N(Q_2) \le \ell_N(Q_3)$, so
			$\ell_{\len(R_1)}(Q_1) = \ell_N(Q_1) \le
			 \ell_N(Q_3) = \ell_{\len(R_3)}(Q_3)$.
		    \item $\vecR_1 \neq \vecR_2 \neq \vecR_3$.			
			Since $\sqsubseteq$ is a partial order, this implies
			in particular that
			$\vecR_1 \sqsubset \vecR_2 \sqsubset \vecR_3$,
			thus $\vecR_1 \sqsubset \vecR_3$ and hence
			$\vecR_1 \neq \vecR_3$.
		    \item $\vecR_1 \neq \vecR_2 = \vecR_3$ or
			$\vecR_1 = \vecR_2 \neq \vecR_3$. In either case,
			$\vecR_1 \neq \vecR_3$.
		\end{enumerate}		
	    \item[$\le_s$ is well-founded:]
		Let $\vecR_1|Q_1 \ge_s \vecR_2|Q_2 \ge_s \cdots$, and set
		$p_i := \phi(\vecR_i|Q_i)$. Then $p_1 \ge_n p_2 \ge_n \cdots$.

		Since $\ge_n$ is well-founded, there is an $i$ such that for
		all $j > i$, $p_j = p_{j+1}$. In particular,
		$p_j \not>_n p_{j+1}$. Thus, $\vecR_j \not>_s \vecR_{j+1}$
		for all $j > i$.
	\end{description}
\qed
    \end{proof}

    \begin{lemma}
	\label{APP-lem:rules-R-order}
	If $\vecR|Q \goesto \vecRI|Q'$ as a result of  applying the
	$\rCandidate$,  $\rDecide$, $\rConflict$, or $\rInduction$ rule, then
	$\vecR|Q >_s \vecRI|Q'$.
    \end{lemma}

    \begin{proof}
	Case analysis on the applied rule.
	\begin{description}
	    \item[\rCandidate:] In this case, $Q = \varnothing$,
		$\vecR = \vecRI$ and $Q' = \{ \langle a, N \rangle \}$ for
		some $a \in \Sigma$. Thus, $\ell_N(Q) = N+1 > \ell_N(Q') = N$.
	    \item[\rDecide:] In this case, $\min Q = \langle a, i \rangle$,
		$\min Q' = \langle b, i-1\rangle$ for some $a,b \in \Sigma$
		and $i > 0$. Also, $\vecR = \vecRI$.

		Thus, $\vecR = \vecRI$ and $\ell_N(Q) = i > i-1 = \ell_N(Q')$.
	    \item[\rConflict:] In this case,
		$\vecRI = \vecR[\setR_k \gets \setR_k \setminus
		\uclosure{b}]_{k=1}^i$ for some $i \ge 1$,
		$b \in \gen_i(a)$, $a \in \Sigma$. By definition of
		$\gen$, we have in particular that $b \in \setR_i$,
		and $b \not\in \setRI_i$. Since furthermore
		$\setRI_j \subseteq \setR_j$ for all $j \le N$, we have
		$\setRI \sqsubset \setR$.
	    \item[\rInduction:] Analogous to \rConflict.
	\end{description}
\qed
    \end{proof}

\begin{prop-non}\textbf{\emph{\ref{prop:inner-loop-terminates}}} [Infinite sequence condition]
For every infinite sequence
	$\init \goesto \sigma_1 \goesto \sigma_2 \goesto \cdots$,
	there are infinitely many $i$ such that $\sigma_i \goesto \sigma_{i+1}$
	by applying the rule \rUnfold.
    \end{prop-non}

    \begin{proof}
	Let $\init \goesto \sigma_1 \goesto \sigma_2 \goesto \cdots$ be an
	infinite sequence of states. Since $\resValid$ and $\resInvalid$ have
        no successor states, all 
	$\sigma_i$ must	be of the form $\vecR_i|Q_i$. Thus, only the following
	rules can be applied to get from $\sigma_i$ to $\sigma_{i+1}$:
	\rUnfold, \rCandidate, \rConflict, \rDecide{} and \rInduction.

	Suppose that there is some $K$ such that for all $i > K$, the transition
	$\sigma_i \goesto \sigma_{i+1}$ is not due to \rUnfold.

	But then, the transition is due to one of \rCandidate, \rConflict,
	\rDecide{} and \rInduction. By Lemma \ref{APP-lem:rules-R-order},
	this means that $\sigma_K >_s \sigma_{K+1} >_s \sigma_{K+2} >_s \cdots$,
	i.e., from $K$ on, the $\sigma_i$ form a $>_s$-descending chain.

	Since the $\sigma_i$ form an infinite sequence, this implies that
	the sequence $\sigma_{K+0}, \sigma_{K+1}, \ldots$ forms an
	infinite $>_s$-chain. But by Lemma \ref{lem:les_wfqo}, $\le_s$
	is wellfounded, so no infinite $>_s$-chains exist -- contradiction.

	Thus, there must be infinitely many $i$ such that
	$\sigma_i \goesto \sigma_{i+1}$ using \rUnfold.
\qed
    \end{proof}

\begin{theorem-non}\textbf{\emph{\ref{lem:no-bad-unfolds}}}
	If there is a path from $I$ to $\Sigma \setminus \setP$ of length
	$k$, the rule $\rUnfold$ can be applied at most $k$ times: for every
	sequence $\init \goesto \sigma_1 \goesto^{*} \sigma_n$,
	there are at most $k$ different values for $i$ such that
	$\sigma_i \goesto \sigma_{i+1}$ using the $\rUnfold$ rule.
\end{theorem-non}

    \begin{proof}
	Let $\init \goesto \sigma_1 \goesto^{*} \sigma_K$ be a sequence
	of rule applications in which has occured $N = k$ times, i.e.,
	there are $i_1 < \cdots < i_k$ such that
	$\sigma_{i_j} \goesto \sigma_{i_j+1}$ via \rUnfold.

	We wish to show that there is no $\sigma'$ such that
	$\sigma_K \goesto \sigma'$ via $\rUnfold$.

	If $\sigma_K \neq \vecR|Q$, the statement follows because $\resValid$ and $\resInvalid$
        have no successors. 
	Thus, consider the case $\sigma_K = \vecR|Q$.

	Let $s_0, \ldots, s_N$ be a path from $I$ to
	$\Sigma \setminus \setP$, i.e., $s_0 \in I$,
	$s_k \in \Sigma \setminus \setP$ and $s_i \to s_{i+1}$ for
	$i = 0, \ldots, N-1$. Then, in particular, $s_i \in \setR_i$ for
	$i = 1, \ldots, N$ by \eqref{eq:I2}.
	
	Thus, the pre-condition for \rUnfold{} is not fulfilled, since
	$s_i \in \setR_N \setminus \setP$.
\qed
    \end{proof}

\begin{lemma}
	\label{APP-prop:finite_di}
	For $i > L$, $D_i = \varnothing$. This implies that the set $\bigcup_{i \ge 0} D_i$ is finite.
    \end{lemma}

    \begin{proof}
	We first show a small auxillary fact:

	\begin{description}
	    \item[Claim:] $D_j \subseteq \setU_j$ for all $j \ge 0$.

	    \item[Proof:] By induction on $j$.
	\begin{itemize}
	    \item $D_0 \subseteq \Sigma \setminus \setP = \setU_0$.
	    \item $D_{j+1} \subseteq \pre{\uclosure{D_j}} \subseteq
		\pre{\setU_j} \subseteq \setU_{j+1}$, using the induction
		hypopthesis in the second step.
	\end{itemize}
	\end{description}

	Now, let $i > L$. By Lemma~\ref{APP-lem:D_iProperties}, statement (2) and the above claim, we have
	$D_i \subseteq \setU_i \setminus \setU_{i-1} =
	\setU_L \setminus \setU_L = \varnothing$, since
	$\setU_j = \setU_L$ for all $j \ge L$ by Lemma 3.4 and the
	discussion in Paragraph 4 of \cite{ACJT96}.

	Since for all $i > L$, $D_i = \varnothing$, it is
	sufficient to show that $D_i$ is finite for $i = 0, ..., L$.
	This is guaranteed by Lemma~\ref{APP-lem:D_iProperties}, statement (5).
\qed
    \end{proof}

\begin{theorem-non}\textbf{\emph{\ref{lm:impl-pred}}}    
        Let $\bfa\in\en^n$ be a state and $t=(\bfg,\bfd)\in\en^n\times\zed^n$ be a transition. Then $\bfb\in\pre{\uclosure{\bfa}}$
        is a predecessor along $t$ if and only if $\bfb\wqoge\max(\bfa-\bfd,\bfg)$.
\end{theorem-non}
    
    \begin{proof}
        Suppose $\bfb\in\pre{\uclosure{\bfa}}$ is a predecessor along $t$. Then $\bfb\wqoge \bfg$ and $\bfb+\bfd\wqoge\bfa$.
        Thus, $\bfb\wqoge\max(\bfa-\bfd,\bfg)$.
        For the other direction, due to well-structuredness it is enough to show $\max(\bfa-\bfd,\bfg)$ itself is a
        predecessor along $t$. But this holds since $\max(\bfa-\bfd,\bfg)\wqoge\bfg$ and
        $\max(\bfa-\bfd,\bfg)+\bfd\wqoge(\bfa-\bfd)+\bfd=\bfa$.
\qed
    \end{proof}

\begin{theorem-non}\textbf{\emph{\ref{lem:generalization-c}}}
        Let $\bfa,\bfc\in\NN^n$ be states and $t=(\bfg,\bfd)\in\NN^n\times\zed^n$ be a transition.
\begin{enumerate}
\item Let $\bfc\wqole\max(\bfa-\bfd,\bfg)$. Define $\bfa''\in\NN^n$ by $a''_j := c_j + d_j$ if $g_j < c_j$
     and $a''_j := 0$ if $g_j \geq c_j$, for $j=1,\ldots, n$.
        % \begin{equation}
        %     \label{eqn:generalization}
        %     a''_j :=
        %     \begin{cases}
        %         c_j + d_j, & g_j<c_j \\
        %         0, & g_j\geq c_j\,,
        %     \end{cases}
        % \end{equation}
        Then $\bfa''\wqole\bfa$. Additionally, for each $\bfa'\in\en^n$ such that
        $\bfa''\wqole\bfa'\wqole\bfa$, we have $\bfc\wqole\max(\bfa'-\bfd,\bfg)$.
\item
        If $\bfa\wqole\max(\bfa-\bfd,\bfg)$, then for each $\bfa'\in\en^n$ such that
        $\bfa'\wqole\bfa$, it holds that $\bfa'\wqole\max(\bfa'-\bfd,\bfg)$.
\end{enumerate}
\end{theorem-non}
    
    \begin{proof}
        For the first part, consider coordinate $j$ for $1\leq j\leq n$. If $g_j\geq c_j$, then
        $a_j'' = 0 \leq a_j$ and $c_j\leq g_j\leq \max(a'_j-d_j,g_j)$.
        
        On the other hand, suppose $g_j<c_j$. First note that $\max(a_j-d_j,g_j)=a_j-d_j$
        since $g_j<c_j\leq\max(a_j-d_j,g_j)$. Thus,
        \[ a_j'' = c_j + d_j \leq \max(a_j - d_j,g_j) + d_j = (a_j - d_j) + d_j = a_j \]
        and
        \[ c_j = a_j'' - d_j \leq a_j' - d_j \leq \max(a'_j-d_j, g_j)\,. \]

For part (2), consider coordinate $j$ for $1\leq j\leq n$. If $a_j-d_j\geq g_j$, then
        \[ a_j\leq \max(a_j-d_j,g_j) = a_j-d_j \,. \]
        Therefore $d_j\leq 0$, implying
        \[ a_j'\leq a_j'-d_j \leq \max(a_j'-d_j,g_j) \,. \]
        On the other hand, if $a_j-d_j<g_j$, then
        \[ a_j'\leq a_j \leq \max(a_j-d_j,g_j)=g_j\leq \max(a_j'-d_j,g_j)\,. \]
\qed
    \end{proof}

\end{document}